\documentclass[a4paper, 11pt]{article}

\usepackage{a4wide}
\usepackage[utf8]{inputenc}
\usepackage[english]{babel} 
\usepackage{amsmath}
\usepackage{amsfonts}
\usepackage{amssymb}
\usepackage{amsthm}
\usepackage{amscd}
\usepackage{euscript}
\usepackage{tikz}
\usepackage{ifthen}
\usetikzlibrary{calc}
\usepackage{url}
\usetikzlibrary{topaths,calc}
\usetikzlibrary{shapes.misc}
\usepackage{subfig}
\usepackage{multicol}

\newtheorem{theo}{Theorem}[section]
\newtheorem{defi}[theo]{Definition}
\newtheorem{prop}[theo]{Proposition}
\newtheorem{coro}[theo]{Corollary}
\newtheorem{lemma}[theo]{Lemma}
\newtheorem{example}[theo]{Example}

\DeclareMathOperator{\Supp}{Supp}

\DeclareMathOperator{\Ker}{Ker}
\DeclareMathOperator{\im}{Im}

\DeclareMathOperator{\Sq}{Sq}
\DeclareMathOperator{\D}{D}

\newcommand{\C}{\mathbb{C}}

\newcommand{\Z}{\mathbb{Z}}

\newcommand{\F}{\mathbb{F}}

\newcommand{\h}{\mathcal{H}}

\newcommand{\Pauli}{\mathcal{P}}

\newcommand\ket[1]{|#1\rangle}

\title{Generalized surface codes and packing of logical qubits}

\makeatletter
\let\@fnsymbol\@arabic
\makeatother

\author{Nicolas Delfosse\thanks{Department of Physics and Astronomy, University of California, Riverside, CA, USA}
$^{,}$\thanks{IQIM, California Institute of Technology, Pasadena, CA, USA}\ \ 
Pavithran Iyer\thanks{Département de Physique and Institut Quantique, Université de Sherbrooke, Québec, Canada}
\ and David Poulin\footnotemark[3]
}

\begin{document}

\maketitle

%\maketitle

\begin{abstract}
We consider a notion of relative homology (and cohomology) for surfaces with two types of boundaries.
Using this tool, we study a generalization of Kitaev's code based on surfaces with mixed boundaries.
This construction includes both Bravyi and Kitaev's  \cite{BK98} and Freedman and Meyer's \cite{FM01} extension 
of Kitaev's toric code.
We argue that our generalization offers a denser storage of quantum information.
In a planar architecture, we obtain a three-fold overhead reduction 
over the standard architecture consisting of a punctured square lattice.
\end{abstract}

\let\thefootnote\relax\footnote{Corresponding author: Nicolas Delfosse - ndelfoss@caltech.edu}

\section{Surface codes}

Kitaev's toric code \cite{Ki03} is one of the most emblematic examples of topological 
quantum codes.
It is defined by local constraints on qubits placed on a torus. The properties of 
the code depend on the topology of the surface. For instance, the number of
encoded qubits is determined by the genus of the surface and the minimum distance 
is the length of the shortest cycle with non-trivial homology.
Such a code can be defined on an arbitrary closed surface.

\medskip
For practical purposes, a planar layout of the qubits is desirable.
Unfortunately, in such a case the prescription of Kitaev to construct quantum codes
yields a trivial code that preserves just a single state.
The work of Freedman and Meyer \cite{FM01} and Bravyi and Kitaev \cite{BK98}
extends Kitaev's construction in two directions.

\medskip
(i) Kitaev construction can be extended to a surface with boundaries, that is a surface
punctured with holes \cite{FM01}.

\medskip
(ii) Two different kinds of boundaries called open and closed boundaries 
(also called rough and smooth respectively) can be introduced along the 
outer boundary of a planar lattice \cite{BK98}.

\medskip
These modifications increase the degeneracy of the ground space of 
the Hamiltonian, allowing for a non-trivial planar surface code.
Punctured planar lattices have been proposed as quantum memory or 
quantum computing architecture where logical operations are realized by braiding holes
\cite{DKLP02, RH07, RHG07, RHG06, WFHL11, FMMC12, HFDV12:surface_code_surgery, BT15:hyperbolic}.
Understanding and optimizing their performance is a central question for the physical implementation of quantum information processing.

\medskip
Surprisingly, mixed boundaries (partially open and partially closed boundaries) have been examined 
previously only for the outer boundary of a planar lattice as in (ii), but not along 
all the punctures. 
In the present work, we combine the two ideas (i) and (ii), producing better surface 
codes.
First, we introduce a family of generalized surface codes based on 
punctured surfaces where any holes can have partially open and partially closed 
boundaries. Justifying the commutation relations between stabilizer generators 
as well as computing the parameters of these codes is non-trivial and requires an in-depth 
study of a notion of relative homology. 
% Moreover, a new notion of dual lattice which encodes the duality between 
% $X$-errors and $Z$-errors has to be introduced.
We determine a closed formula for the parameters of these generalized surface codes 
and we describe graphically their logical operators.
Then, we propose a planar architecture based on holes with mixed boundaries 
which improves over the parameters of standard constructions of 
two-dimensional surface codes. We obtain a three-fold reduction of the overhead compared
to the square lattice punctured with closed holes \cite{FMMC12}.

\medskip
Besides providing new constructions of surface codes, the formalism developed here 
is necessary in order to optimize the performance of different fault-tolerant architecture.
If it is clear that surface codes can make a quantum computer fault-tolerant, the details of the
architecture of such a fault-tolerant quantum computer are still to be determined. Optimizing the 
design of surface codes can lead to a much more favourable overhead. Our formalism, which 
encompasses all the previously considered constructions of surface codes
\cite{Ki03, FM01, BK98, DKLP02, RH07, RHG07, RHG06, WFHL11, FMMC12, Ze09, HFDV12:surface_code_surgery, BT15:hyperbolic}, 
provides an ideal framework to compare and optimize quantum computing architectures 
based on surface codes.

\medskip
Generalized surface codes are defined in Section~\ref{section:codes}. 
The definition of these codes and the computation of their parameters rely on a particular 
notion of relative homology of surfaces with boundaries that we study 
in Section~\ref{section:homology}. Finally, a planar architecture based on generalized surface 
codes is proposed in Section~\ref{section:planar_codes}.
This last section, which focuses on the problem of optimizing the packing of logical qubits in a 
planar lattice, can be read independently.

\section{Background on stabilizer codes}

Let us recall the definition of {\em stabilizer codes} \cite{Go97}.
In what follows, $I, X, Y$ and $Z$ denote the usual Pauli matrices. 
Pauli operators are $n$-fold tensor products of Pauli matrices 
$
i^a P_1 \otimes P_2 \otimes \dots \otimes P_n
$
where $a \in \Z_4$. Denote by $\Pauli_n$ the set of
$n$-qubit Pauli operators.

\medskip
A stabilizer code of length $n$ is defined as the common $(+1)$-eigenspace 
of a family of $n$-qubit commuting Pauli operators $S_1, S_2, \dots S_r$.
Equivalently, it is the degenerate ground space of the Hamiltonian $H = -\sum_i S_i$.
It is a $2^{n-r}$-dimensional subspace of $(\C^2)^{\otimes n}$, whenever
the $r$ Pauli operators $S_i$ are independent. The quantum code then encodes
$k=n-r$ qubits into $n$ qubits and its parameters are denoted as $[[n, k]]$. The 
group generated by the operators $S_i$ is denoted $\mathcal{S}$ and is called the stabilizer 
group of the quantum code $C(\mathcal S)$.

\medskip
Assume that an encoded state $\ket \psi$ is subjected to a Pauli
error $E \in \Pauli_n$. The system is then in the state $E \ket \psi$.
The correction procedure for stabilizer codes is based on 
the {\em syndrome measurement}, that is the measurement of the 
observables $S_i$ for all $i=1,\dots, r$.
These commuting operators can be measured simultaneously, providing a 
measurement outcome $\pm 1$ for all $i$. The outcome of the 
measurement of a stabilizer $S_i$ is $(-1)^{\sigma_i}$ where 
$\sigma_i \in \F_2$ is defined by the equation 
$E S_i = (-1)^{\sigma_i} S_i E$. This defines the syndrome 
$\sigma(E) = (\sigma_1, \dots, \sigma_r) \in \F_2^r$.
Whenever the system being measured is in a state $\ket \psi$ that 
belongs to the code space the syndrome is trivial. Hence, a non-trivial 
syndrome indicates the presence of an error.

\medskip
Among stabilizer codes, {\em CSS codes} \cite{CS96, St96} are those defined by $r_X$ operators 
chosen from $\{I, X\}^{\otimes n}$ and $r_Z = r-r_X$ 
operators chosen from $\{I, Z\}^{\otimes n}$. Writing Pauli errors as $E = i^a E_Z E_X$ 
with $E_Z \in \{I, Z\}^{\otimes n}$ and $E_X \in \{I, X\}^{\otimes n}$,
we see that the syndrome can be partitioned as a pair of vectors 
$(\sigma_Z, \sigma_X)$, where $\sigma_Z$ (resp. $\sigma_X$) contains the 
measurement outcomes of the $r_X$ operators of $X$-type (resp. the $r_Z$ operators of 
$Z$-type). As suggested by the notation, the vector $\sigma_Z$ depends
only on $E_Z$ and $\sigma_X$ depends only on $E_X$. Since quantum states 
are defined up to a phase, we can assume that the phase $i^a$ of the error is 
trivial. Our goal will be to identify $E_X$ from its syndrome $\sigma_X$
and $E_Z$ from its syndrome $\sigma_Z$.

\medskip
An error which has a trivial syndrome is called a {\em logical operator} 
or a {\em logical error}. These errors preserve the code space.
A stabilizer, that is an element of the stabilizer group, 
is a particular logical operator that has a trivial action on 
encoded qubits. By a {\em non-trivial logical operator}, we mean 
a logical operator that is not a stabilizer.
Up to the stabilizer group $S$, the set of logical operators, denoted $N(S)$ 
(for normalizer), is a group that has the same structure 
as the $k$-qubit Pauli group. More precisely, the quotient $N(S)/S$ is generated 
by $2k$ operators $\bar X_1, \bar Z_1 \dots, \bar X_k, \bar Z_k$ which satisfy 
the same relations as the $k$-qubit Pauli operators $[\bar P_i, \bar Q_j] = [P_i, Q_j]$. 
A family which satisfies these relations is called the {\em symplectic basis} of the 
logical operators.

\medskip
The {\em minimum distance} $d$ of a stabilizer code is defined to be the 
minimum weight of a non-trivial logical error.
It is a proxy indicator for the error-correction capability of the code. 
When $d$ is known, it is added to the parameters of the code denoted as 
$[[n, k, d]]$.
For CSS codes, the minimum distance is reached either by an error 
$E_Z \in \{I, Z\}^{\otimes n}$ or by an error $E_X \in \{I, X\}^{\otimes n}$. 
One can thus obtain the minimum  distance as $d = \min\{d_X, d_Z\}$ 
where $d_Z$ is the minimum weight of a non-trivial logical error 
$E_Z \in \{I, Z\}^{\otimes n}$ and $d_X$ is the 
minimum weight of a non-trivial logical error $E_X \in \{I, X\}^{\otimes n}$.

\section{Definition and properties of generalized surface codes} \label{section:codes}

Stabilizer codes and CSS codes can be considered as a quantum generalization of
classical linear codes. The main obstacle which arises when we try to 'quantize' a classical 
code construction is the constraint that stabilizers must commute.
In Kitaev's construction, qubits are placed on the edges of a cellulation of a 
surface, $X$-type stabilizers correspond to vertices and $Z$-type stabilizers 
correspond to faces. 
The commutation relations are then an immediate consequence of basic homological properties
(the boundary of a boundary is trivial). The homology of the surface is also a 
crucial tool to compute the parameters of Kitaev's code.

In this section, we first introduce the surfaces that support our code construction.
Then, we define generalized surface codes and we state our main result which provides
a full description of the parameters of these codes. The validity of the construction as well as the proof 
of our main result is delayed to Section~\ref{section:homology}  through
the study of relative homology.

\subsection{Surface with open and closed boundaries}

The goal of this section is to define an appropriate notion of surface in order 
to incorporate all extensions of Kitaev's code.
Starting from closed surfaces, we bring two modifications. First, we authorize
surfaces with boundaries. Then, each edge on the boundary will be declared 
either open or closed. We choose the term open or closed for its topological
connotation. A surface with only closed boundaries is a closed surface.
Those boundaries are sometimes called rough and smooth in quantum information.
This notion of surface will allow us to consider a general construction of surface 
codes which encompasses both Kitaev's original construction \cite{Ki03},
its generalizations \cite{BK98, FM01} and extends it to any surface 
punctured with mixed holes that have open as well as closed edges on their boundaries, 
as one can see in Figure~\ref{fig:surfaces}.

\begin{figure}[h]
\begin{center}
\includegraphics[scale=.7]{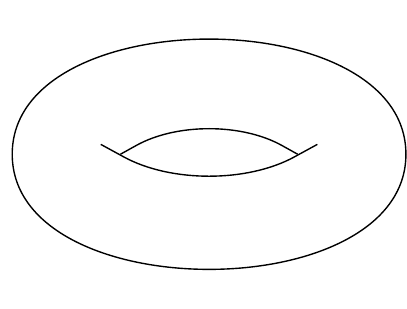}
\hspace{.2cm}
\includegraphics[scale=.6]{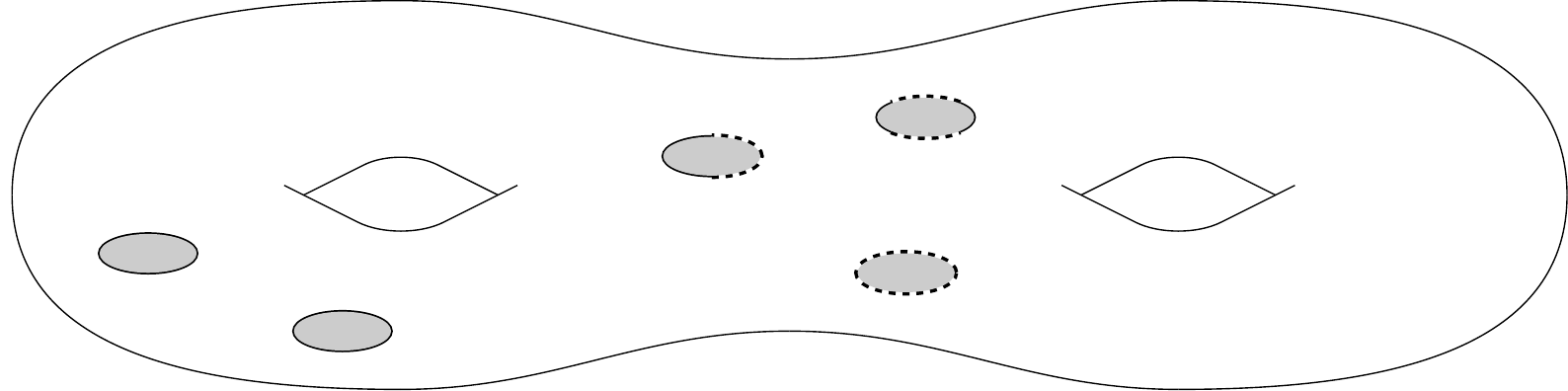}
\caption{Left: A closed surface of genus 1. Right: 
A surface of genus 2 with boundaries. The surface is punctured with 5 holes resented by 
dark disks. The boundary of these holes can be open (dotted lines)
or closed (continuous lines).
Here 2 holes are closed and 1 hole is open. The 2 remaining holes have both open
and closed boundaries.}
\label{fig:surfaces}
\end{center}
\end{figure}

\medskip
Since it is only the combinatorial structure of the surface that plays a role in
the definition of topological codes, we consider surfaces constructed
by gluing a finite set of faces together along their edges in a coherent
way. Formally, a {\em combinatorial surface} or a simply a {\em surface} 
is a cellulation of a compact 2-manifold $S$ (possibly with boundaries).
In other words, it is a triple $(V, E, F)$ where $G = (V, E)$ is a finite graph embedded 
on the surface such that each connected component of $S \backslash G$
is a disk which is a face $f \in F$ of the cellulation.
We often identify a face $f$ with the set of edges on its boundary, 
\emph{i.e.} we regard a face $f$ as a subset $f \subset E$. Topologically a 
face is a disk and its boundary is a set of edges that 
form a cycle in the graph $(V, E)$.
We assume that $G$ is embedded without overlapping edges or multiple edges.
We also suppose that no edge belongs to the same face twice and that two faces 
share at most one edge.

\medskip
Now that we have a surface $G=(V, E, F)$, let us define its {\em boundary}.
An edge of $G$ is defined to be a boundary if it is incident to a single face of $G$. 
The corresponding faces are called boundary faces. The endpoints of boundary 
edges are the boundary vertices.
We denote the sets of boundary vertices, edges and faces, 
respectively by $\partial V, \partial E$ and $\partial F$.

\medskip
Boundary elements will be either {\em open} or {\em closed}.
First, a subset of edges living on the boundary of the tiling is 
declared to be open. Once open edges are defined, a vertex 
of the boundary is said to be open if it is incident to an open 
edge. Analogously an open face is a face of the boundary 
containing an open edge.
An edge, respectively a vertex or a face of the boundary that 
is not open is said to be closed. This defines a partition of the 
boundary 
$\partial V = \partial_O V \sqcup \partial_C V$,
$\partial E = \partial_O E \sqcup \partial_C E$,
and 
$\partial F = \partial_O F \sqcup \partial_C F$,
where $\partial_O$ denotes the open subset and 
$\partial_C$ denotes the closed subset.

\medskip
The set of non-open vertices, edges and faces are denoted
respectively by $\mathring{V} = V \backslash \partial_O V$, 
$\mathring{E} = E \backslash \partial_O E$ and
$\mathring{F} = F \backslash \partial_O F$.
By analogy with topology, one may be tempted to define the 
set $\mathring{X}$ as $X \backslash \partial X$. Our definition
makes the statements of our results simpler and emphasizes the similarities 
with standard results in homology.

\subsection{Generalized surface codes}

In this section we provide a unified definition for surface codes defined on
surfaces with or without boundaries.

In order to construct a quantum error-correcting code from a surface $G=(V, E, F)$, 
we place a qubit on each non-open edge of $G$. This leads to a global state
living in the Hilbert space $\h = \otimes_{e \in \mathring{E}} \C^2$.
Denote by $X_e$, respectively $Z_e$, the Pauli operator acting as $X$, respectively
$Z$, on the qubit indexed by $e$ and which is the identity on the other qubits.
 
\begin{defi} \label{defi:surface_codes}
The {\em surface code} associated with the surface $G=(V, E, F)$ is defined to be the 
stabilizer code of stabilizer group $S = \langle X_v, Z_f \ | \ v \in \mathring{V}, f \in F \rangle$,
where 
$$
X_v = \prod_{\substack{e \in \mathring{E} \\ v \in e}} X_e \quad \text{ and } \quad Z_f = \prod_{\substack{e \in \mathring{E} \\ e \in f}} Z_e.
$$
\end{defi}

In other words, this surface code is the ground space of the Hamiltonian 
$$
H = - \sum_{v \in \mathring V} X_v - \sum_{f \in F} Z_f \cdot
$$
The commutation between the operators $X_v$ and $Z_f$, 
for all $v \in \mathring V$ and for all $f \in F$, is ensured 
by Lemma~\ref{lemma:boundary_composition}.

\medskip
When the surface $G$ has no boundaries, this definition coincides with 
Kitaev's original construction \cite{Ki03}. When the surface has only closed boundaries
it is Freedman and Meyer's generalization \cite{FM01}. Finally, 
we recover Bravyi and Kitaev's construction \cite{BK98} when the surface is a sphere punctured 
with a single hole.

\medskip
Our first objective is to establish a closed formula for the parameters of 
generalized surface codes. Let us recall the key ingredients in the case of closed connected surfaces.
\begin{itemize}
\item The rank of the $Z$-stabilizer group 
$S_Z = \lbrace Z_f \ | \ f \in F \rbrace$ 
is given by $|F| - 1$.
\item A $Z$-error $E_Z \in \{I, Z\}^{\otimes n}$ has trivial syndrome if and only if
its support is a cycle and it is a stabilizer if and only if this cycle is
homologically trivial. This proves that the minimum distance $d_Z$ is 
the minimum length of a cycle with non-trivial homology.
\item Replacing the graph $G$ by its dual exchanges the role of $X$ and $Z$,
proving that $d_X$ is the minimum length of a cycle of the dual graph $G^*$ 
with non-trivial homology and that the $X$-stabilizer group has rank $|F^*|-1 = |V|-1$
\end{itemize}
Denote by $\chi(G)$ the Euler-Poincaré characteristic of the surface.
Based on the first and the third items, $k = 2 - \chi(G)$ which is equal to $2g$ when the
surface is orientable and $g$ when it is not orientable.
The minimum distance is given by the last two items as the minimum length of 
a non-trivial cycle of $G$ or its dual $G^*$.

\medskip
In the presence of open and closed boundaries, these 3 items are altered. First, the rank
of the stabilizer group depends on the boundaries in a non-trivial way 
(See Theorem~\ref{theo:code_parameters} below). Second, one needs an appropriate notion
of homology that reproduces the behaviour of errors and syndrome measurements.
Lastly, one also needs to define a dual graph $G^*$ in such a way that graph duality coincides exactly 
with the duality between $X$ and $Z$-errors.

\subsection{Statement of the main result} \label{section:main_result}

An important part of the present article is devoted to the study of the homology
of surfaces with open and closed boundaries and to the construction of an appropriate
notion of dual for these surfaces.
In order to state our main result, we first provide a rough definition of these notions. 
Rigorous definitions of these objects will be given in Section~\ref{section:homology}.

\begin{figure}
\begin{center}
\includegraphics[scale=.35]{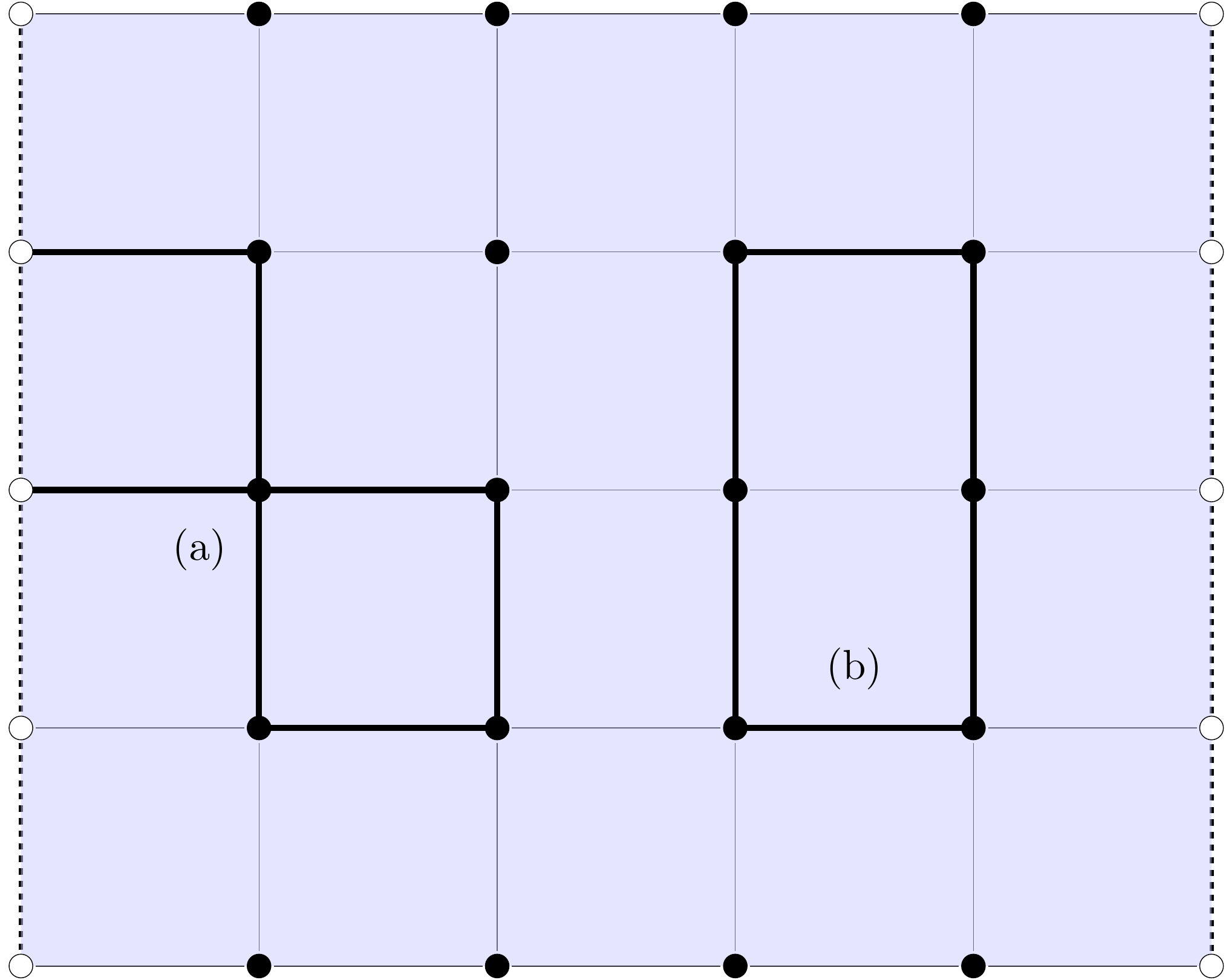}
\hspace{.2cm}
\includegraphics[scale=.3]{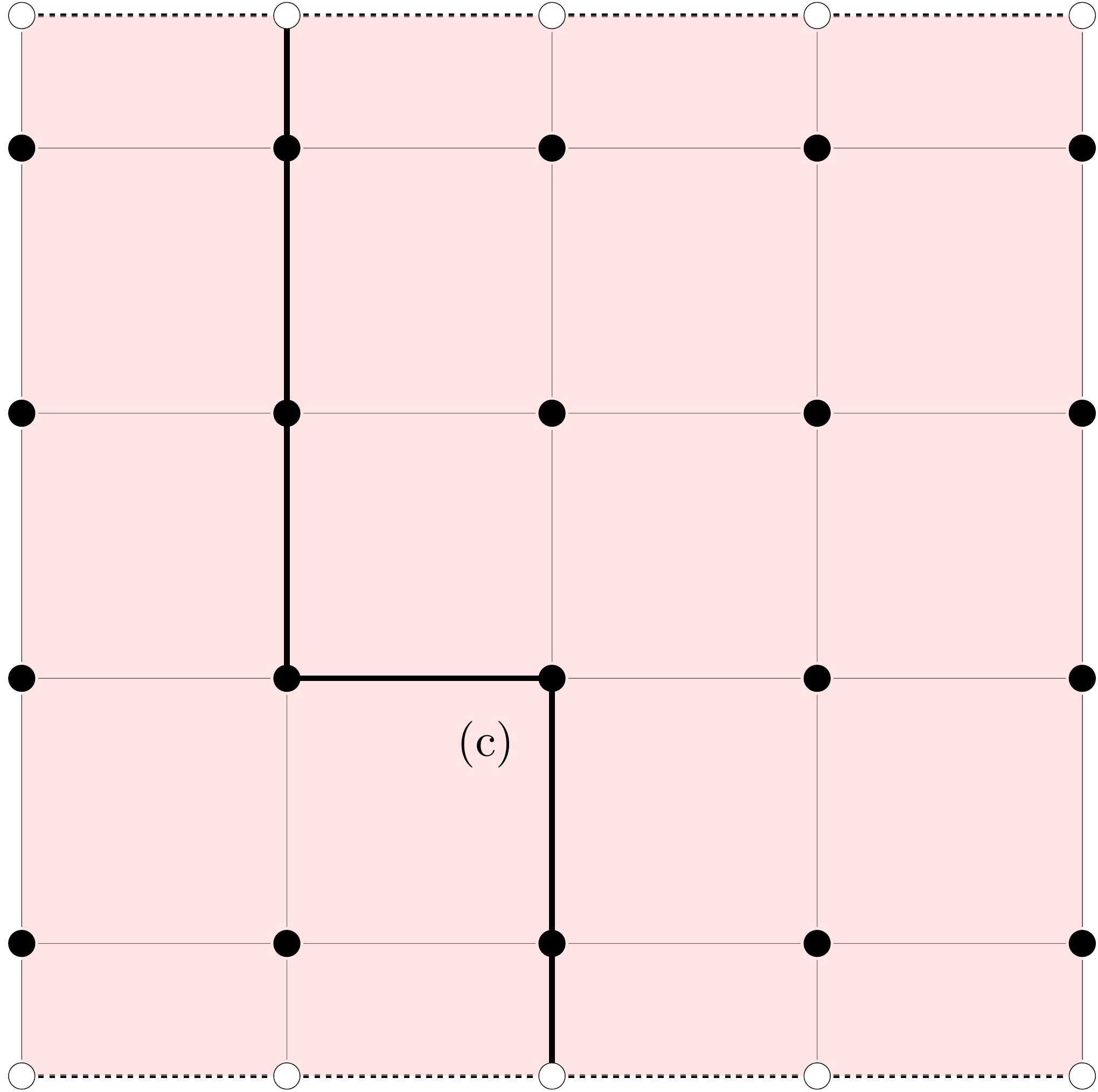}
\caption{A surface with 2 open boundaries and 2 closed boundaries and its dual.
Open edges are represented by dotted lines and open vertices by white nodes. 
Duality switches open and closed boundaries.
Three cycles are depicted. The cycles (a) and (b) are trivial relative cycles whereas the 
relative cycle (c) is non-trivial.}
\label{fig:rect_open}
\end{center}
\end{figure}

\medskip
In order to describe the minimum distance of the code, we need to introduce a 
generalization of the notion of cycle for surfaces with boundaries.
Following Bravyi and Kitaev, we define a {\em relative cycle} of a surface with
boundaries as a set of non-open edges which meet each non-open vertex an even number of times.
Three relative cycles are shown in Figure~\ref{fig:rect_open}.
For instance, a path connecting two open vertices, like (a) and (c) in the figure, 
is a relative cycle though it is not a standard cycle.
The support of an error $E_Z$ with syndrome zero is a relative cycle since it can only
be detected by measuring the operators $X_v$ associated with the non-open 
vertices $v \in \mathring V$. 

\medskip
A relative cycle is said to be {\em homologically trivial} or simply {\em trivial} if it is
the boundary of a subset of faces of the surface. Examples of trivial and non-trivial
cycles are depicted in Figure~\ref{fig:rect_open}.
By definition of generalized surface codes, 
a trivial cycle induces an error $E_Z$ which is a stabilizer (again this cycle is the support of the error).

\medskip
We will also have to extend the definition of the dual to surfaces with open and closed 
boundaries. 
Inspired by the standard example of the square lattice with open 
and closed boundaries depicted in Figure~\ref{fig:rect_open}~\cite{BK98}, the dual graph will 
be defined in such a way that open and closed boundaries are switched under duality. 
This duality will be introduced and motivated in section~\ref{section:dual}.

\medskip
Our first result is a closed formula for the parameters of generalized surface codes
defined on any surface, orientable or not, with boundaries or not, and where boundaries
can be either open or closed.
We use the notation $\kappa_{\overline{\partial_C E}}(G)$ for the number of connected 
components of $G$ containing no closed boundary edge $e \in \partial_C E$ and the notation 
$\kappa_{\overline{\partial_O V}}(G)$ for the number of connected components 
of $G$ containing no open vertex $v \in \partial_O V$.
\begin{theo} \label{theo:code_parameters}
The parameters $[[n, k, d]]$ of the generalized surface code associated with a surface $G=(V, E, F)$ are
\begin{itemize}
\item $n = |\mathring{E}|$, 
\item $k = - |\mathring{V}| + |\mathring{E}| - |F| + \kappa_{\overline{\partial_O V}}(G)  + \kappa_{\overline{\partial_C E}}(G),$
\item $d = \min \{d_Z, d_X\}$ where $d_Z$ is the minimum length of a non-trivial relative cycle of $G$
and $d_X$ is the minimum length of a non-trivial relative cycle of $G^*$.
\end{itemize}
\end{theo}
The first result is obvious.
Note that in the last item both the dual graph $G^*$ and the notion of cycle and trivial cycle 
are generalized to the case of surfaces with open and closed boundaries.
This theorem will be proven through the study of the homology group of surface 
with open and closed boundaries,
%denoted $H_1^\partial(G)$,
which is the focus of Section~\ref{section:homology}. 
It is the combined conclusion of Corollary~\ref{coro:k} and Corollary~\ref{coro:d}.

\medskip
{\bf Remark:} In order to avoid surface codes with minimum distance $d=1$, one can assume
that any edge whose both endpoints are open is an open edge. See the remark at the end of
Section~\ref{section:dual}.

\begin{figure}
\centering
\includegraphics[scale=.3]{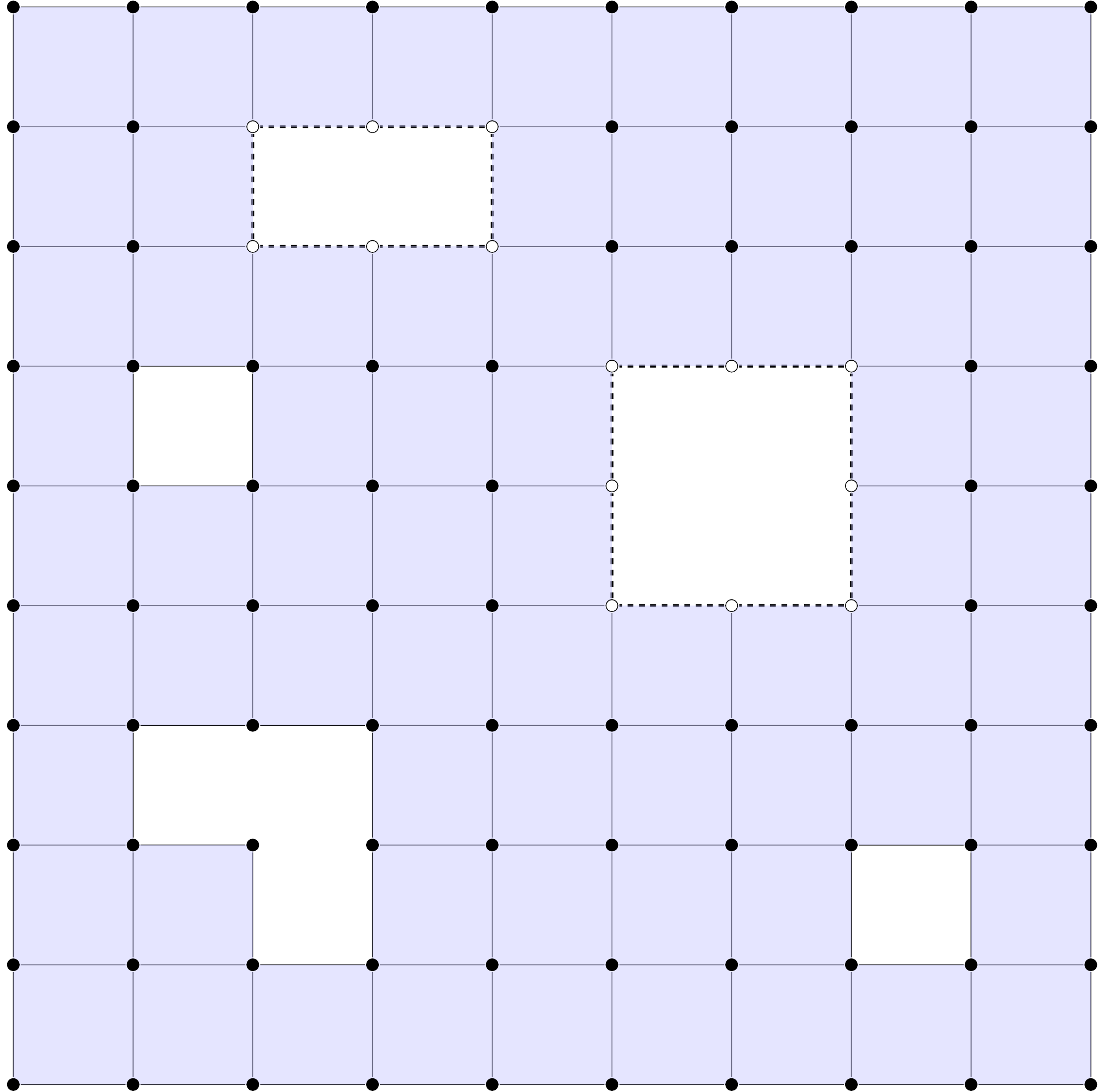}
\includegraphics[scale=.4]{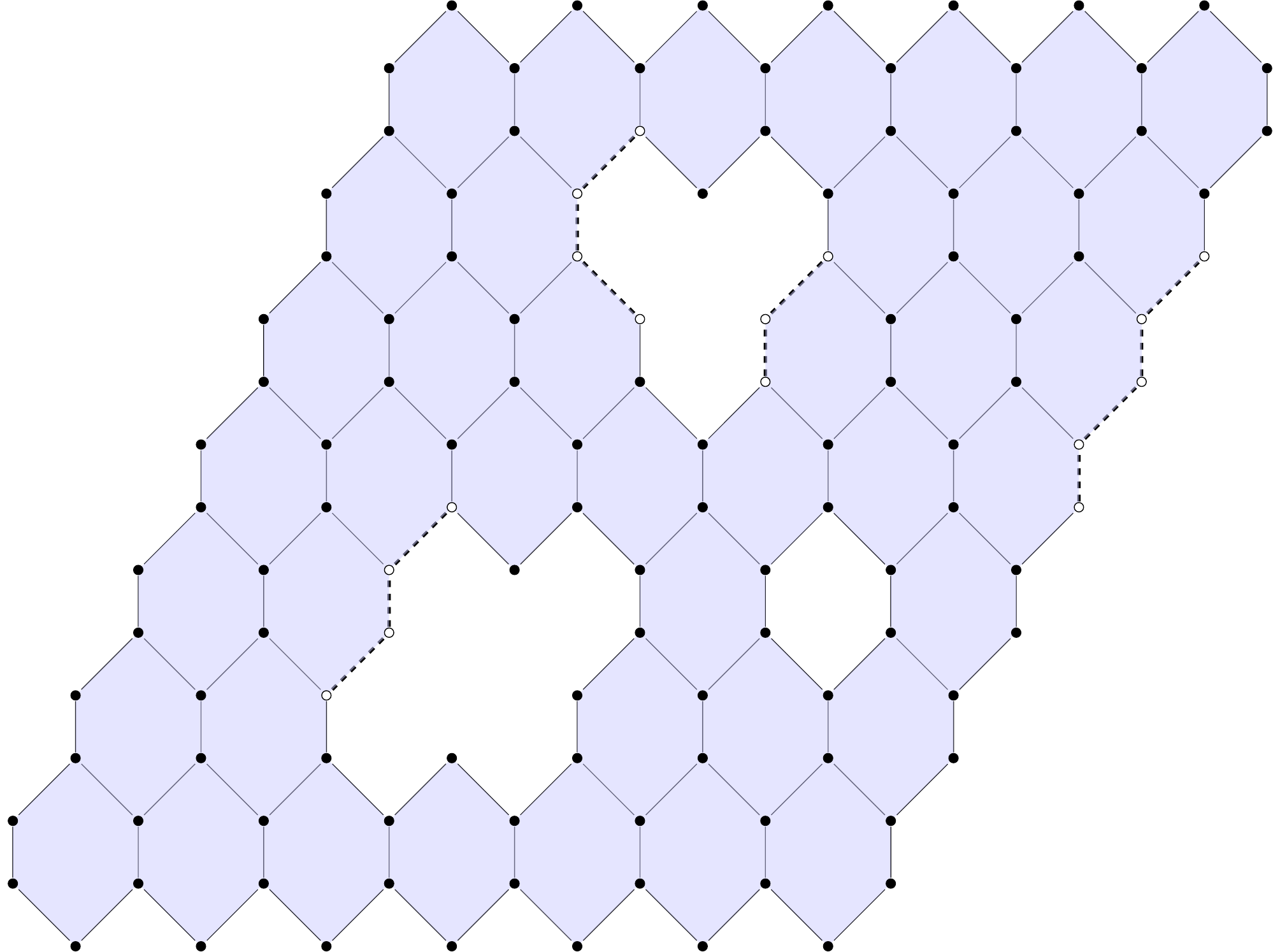}
\caption{Left: A surface of genus 0 with 6 holes encoding k=4 qubits. Right: A surface with partially open holes encoding 6 qubits. Dashed edges represent open edges.}
\label{fig:open_closed_holes}
\end{figure}

\subsection{Applications to surfaces with uniform boundaries}

As a first application, we find the parameters of surface codes defined over
a surface punctured with open holes and closed holes.

\begin{coro} \label{cor:k_uniform_surfaces}
Let  $G$ be a connected surface with $b_c$ closed-boundary holes  and $b_o$ open-boundary holes.
 The number of logical qubits encoded in the corresponding surface code is
$$
\begin{cases}
k = 2g + \min\{b_c -1, 0\} + \min\{b_o-1, 0\} \quad \text{ if $G$ is orientable},\\
k = g + \min\{b_c -1, 0\} + \min\{b_o-1, 0\} \quad \text{ if $G$ is not orientable.}
\end{cases}
$$
\end{coro}

As an illustration, a surface of genus 0 with 4 closed holes (the external boundary is a closed hole)
and 2 open holes encoding $k=4$ qubits is represented in Figure~\ref{fig:open_closed_holes}.

\begin{proof}
If $G$ is an orientable surface of genus $g$ with no boundary, then 
its Euler-Poincaré characteristic $\chi(G) = |V| - |E| + |F|$ equals $2 - 2g$. 
For a non-orientable surface, we have $\chi(G) = 2 - g$. 
The case of closed surface is thus a straightforward application of 
Theorem~\ref{theo:code_parameters}.
If $G$ is connected and contains $b_c$ closed boundaries, then 
in the orientable case, $\chi(G) = 2g-b_c$, which in turn yields
$k = 2g - b_c - 1$ according to Theorem~\ref{theo:code_parameters}.
In the non-orientable case, $k = g-b_c-1$ again according to Theorem~\ref{theo:code_parameters}.

The case of a surface with some open boundaries is less standard.
Assume that $b_o + b_c > 0$.
Denote $\bar G$ the surface obtained from $G$ by closing all the
boundaries. Any open edge is now declared to be closed. This allows us to apply the previous
result to $\bar G$. Let us now compare $k(G)$ and $k(\bar G)$:
$$
k(G) - k(\bar G) = |V| - |\mathring{V}| +  |\mathring{E}| - |E| + \kappa_{\overline{\partial_O V}}(G)  + \kappa_{\overline{\partial_C E}}(G) - \kappa_{\overline{\partial_O V}}(\bar G)  - \kappa_{\overline{\partial_C E}}(\bar G)
$$
Any boundary is either an open cycle or a closed cycle. This implies that the number of 
open edges $|E| - |\mathring{E}|$ is equal to the number of open vertices $|V| - |\mathring{V}|$.

The terms corresponding to closed boundaries add up to
$$
\kappa_{\overline{\partial_C E}}(G) - \kappa_{\overline{\partial_C E}}(\bar G) =
\begin{cases}
%1- 0  = 
1
\quad &\text{ if } b_c = 0\\
%0 - 0 = 
0
\quad &\text{ if } b_c > 0\\
\end{cases}
$$
and the contribution of open boundaries is
$$
\kappa_{\overline{\partial_O V}}(G) - \kappa_{\overline{\partial_O V}}(\bar G) =
\begin{cases}
%1 - 1 = 
0
\quad &\text{ if } b_o = 0\\
%0 - 1 = 
-1
\quad &\text{ if } b_o > 0\\
\end{cases}
$$
Overall, we obtain
$
k(G) = k(\bar G) + \delta_{b_c = 0} - \delta_{b_o > 0}.
$
which is $2g + b - 1 + \delta_{b_c = 0} - \delta_{b_o > 0}$ or 
$g + b - 1 + \delta_{b_c = 0} - \delta_{b_o > 0}$ depending on the 
orientability of the surface. This proves that the result holds.
\end{proof}

\begin{figure}[h]
\begin{center}
\includegraphics[scale=1]{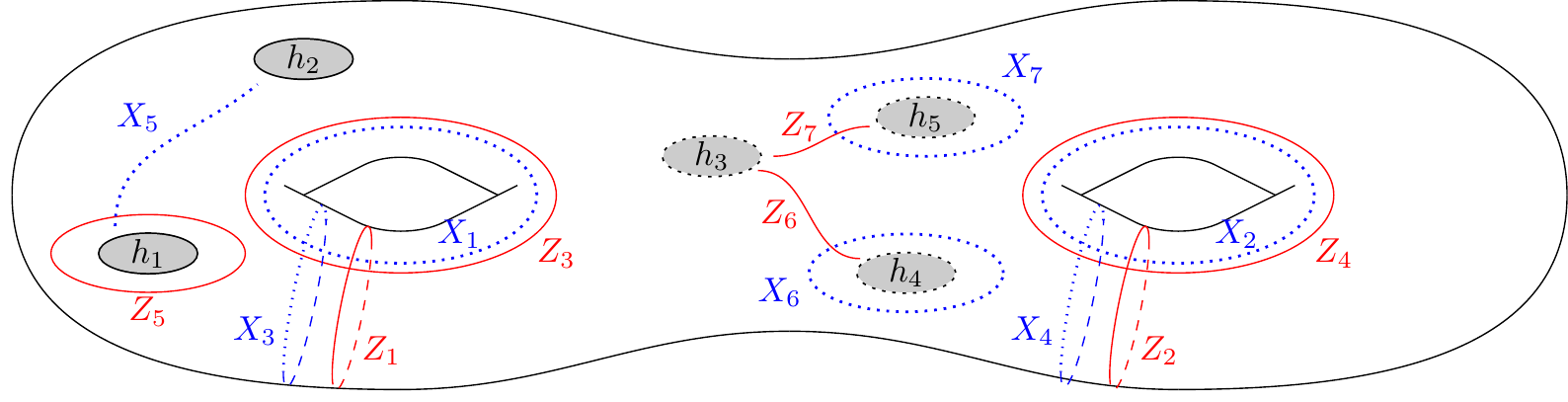}
\caption{A symplectic basis of the set of logical operators for an orientable surface of genus 2
punctured with 5 holes.}
\label{fig:logical_punctured_2torus}
\end{center}
\end{figure} 

For instance, Figure~\ref{fig:logical_punctured_2torus} represents an orientable surface of 
genus 2 punctured by 2 closed holes $h_1$ and $h_2$ and 3 open holes $h_3, h_4$ and $h_5$, 
resulting in a code that encodes $k=4+1+2 = 7$ logical qubits.
A symplectic basis $\bar X_1, \bar Z_1, \dots, \bar X_k, \bar Z_k$ 
of logical operators is shown in the above figure.
Let us detail the construction of this symplectic basis.
As usual, for each handle we define 2 pairs of logical $\bar X_i, \bar Z_i$ of logical operators.
When the surface is not orientable, it is a connected sum of $g$ 
projective planes and we have only one logical of each type per projective
plane.
It remains to define the logical operators associated with the $b_c$ closed holes
and the $b_o$ open holes.
Consider the closed holes $h_1, \dots, h_{b_c}$. The last hole $h_{b_c}$ 
(the top closed hole $h_2$ in Figure~\ref{fig:logical_punctured_2torus}) will play 
a special role.
For the $b_c-1$ holes $h_1, \dots, h_{b_c-1}$, define a $Z$-logical
operator $\bar Z_i$ whose support is a loop $\alpha_i$ around the cycle $h_i$ for $i=1, \dots, b_c - 1$. 
Then, take a family of paths $\beta_1, \dots, \beta_{b_c-1}$ in the dual graph $G^*$ such that 
$\beta_i$ connects the hole $h_{b_c}$ to the hole $h_i$ and 
$\beta_i \cap \alpha_j = \delta_{i, j}$ for all $i, j$. 
By $\beta_i \cap \alpha_j = \delta_{i, j}$,
we mean that the only path $\alpha_j$ which shares an edge with $\beta_i$
(up to duality) is $\alpha_i$, and moreover, these paths have exactly one common edge.
Each of these $b_c-1$ paths $\beta_i$, define a logical operator $\bar X_i$ whose support 
is $\beta_i$. The assumption $\beta_i \cap \alpha_j = \delta_{i, j}$ ensures the 
commutation relations required for a symplectic basis. 
Then, we construct $b_o - 1$ pairs of logical operators based on the $b_o$ open holes.
This can be done by the same procedure in the dual graph.

\subsection{Applications to surfaces with mixed boundaries}

We now consider codes based on surface punctured with holes which 
may have mixed (open and closed) boundaries.

\begin{coro} \label{coro:k_mixed_surface_codes}
If $G$ is a connected surface of genus $g$ with $b$ holes 
that contains $m > 1$ non-cyclic disjoint open paths
then the number of logical qubits of the corresponding surface code is
$$
\begin{cases}
k = 2g + b + m -2 \text{ if $G$ is orientable},\\
k = g + b + m -2 \quad \text{ if $G$ is not orientable.}
\end{cases}
$$
\end{coro}

By non-cyclic path, we simply mean a path that is not a cycle. Since these
paths live on the boundary of a hole, the only kind of cyclic open path that may
appear in this context is a totally open hole.

A similar formula was given for the number of logical qubits of colour codes 
\cite{KYP15}.

\begin{proof}
We proceed as in the proof of Corollary~\ref{cor:k_uniform_surfaces}. 
In order to determine the number of logical qubits $k(G)$ of the code based on the 
surface $G$, we introduce the surface $\bar G$ which is obtained from $G$ by closing
all its boundaries. The open subsets $\mathring V$ and $\mathring E$ of the surface 
$\bar G$ are then trivial.
We already know that $k(\bar G) = 2g + b$ or $g+b$, let us relate $k(G)$ and 
$k(\bar G)$.

Assume that $G$ is punctured by $b$ holes indexed by $i=1, \dots, b$.
The perimeter of the $i$-th hole is a cycle $\gamma_i$ of length $\ell_i$. If some of
its boundaries are open, it is partitioned as an alternate sequence of $2m_i$ open and 
closed path $\gamma_i = p_{i,1} \cup p_{i,2} \cup \dots \cup p_{i,2m_i}$ where odd 
paths are open and even paths are closed.
When a hole $\gamma_i$ is totally closed or totally open, its partition is trivial and we set 
$m_i = 0$.
%By Theorem~\ref{theo:code_parameters}, $k(G)$ is given by
%$$
%k(G) = - |\mathring{V}| + |\mathring{E}| - |F| + \kappa_{\overline{\partial_O V}}(G)  + \kappa_{\overline{\partial_C E}}(G)
%$$
The case of a surface containing only open boundaries was already considered in 
Corollary~\ref{cor:k_uniform_surfaces}. We focus on surfaces $G$ that contain
both types of boundaries.
By Theorem~\ref{theo:code_parameters}, we have
$$
k(G) - k(\bar G) = |V| - |\mathring{V}| +  |\mathring{E}| - |E| + \kappa_{\overline{\partial_O V}}(G)  + \kappa_{\overline{\partial_C E}}(G) - \kappa_{\overline{\partial_O V}}(\bar G)  - \kappa_{\overline{\partial_C E}}(\bar G)
$$
By definition, $\kappa_{\overline{\partial_O V}}(\bar G) = 1$ and $\kappa_{\overline{\partial_C E}}(\bar G) = 0$
whereas $\kappa_{\overline{\partial_O V}}(G) = 0$ and $\kappa_{\overline{\partial_C E}}(G) = 0$
which produces
\begin{align*}
k(G) - k(\bar G) 
& = |V| - |\mathring{V}| +  |\mathring{E}| - |E| - 1\\
& = |\partial_O V| -  |\partial_O E| - 1\\
& = \left( \sum_{i=1}^b \left( |\partial_O V(\gamma_i)| -  |\partial_O E(\gamma_i)| \right) \right) - 1\\
& = \left( \sum_{i=1}^b m_i \right) - 1
\end{align*}
Therein, we simply used the fact that each open path that is not a cycle 
%(which happens only when $m_i>0$) 
contains 1 more open vertex than the number of open edges.
Altogether, we proved that
$$
k(G) = k(\bar G) + \sum_{i=1}^b m_i - 1
$$
where $k(\bar G) = 2g  +b-1$ or $g+b-1$ and $\sum_{i=1}^b m_i$ is the number of non-cyclic disjoint 
open paths on the boundary of $G$.
\end{proof}

For instance, consider a planar lattice punctured with $b$ holes (that is a sphere
punctured by $b+1$ holes).
If the $b$ holes are closed, then the corresponding surface code encodes $k=b$ qubits.
But if the boundary of each of these $b$ holes is the union of an open path and a closed path, 
this generalized surface code encodes $k=2b-1$ qubits that is
approximately twice more than the standard surface codes based on closed holes.
However, one cannot immediately conclude that these codes are always better since 
this transformation may also reduce the minimum distance.

\medskip
A second example is shown in Figure~\ref{fig:logical_mixed_punctured_2torus}. It is an 
orientable surface of genus 2 with 4 mixed boundaries holes. 
One of these 4 holes is open, another one is closed and each of the other 2 holes have 2 open 
paths and 2 closed paths on their boundary.
Corollary~\ref{coro:k_mixed_surface_codes} tells us that the associated surface code
encodes $10$ qubits ($g=2, b=4$ and $m=4$). A symplectic basis of the logical operators
is represented.

\begin{figure}[h]
\begin{center}
\includegraphics[scale=1]{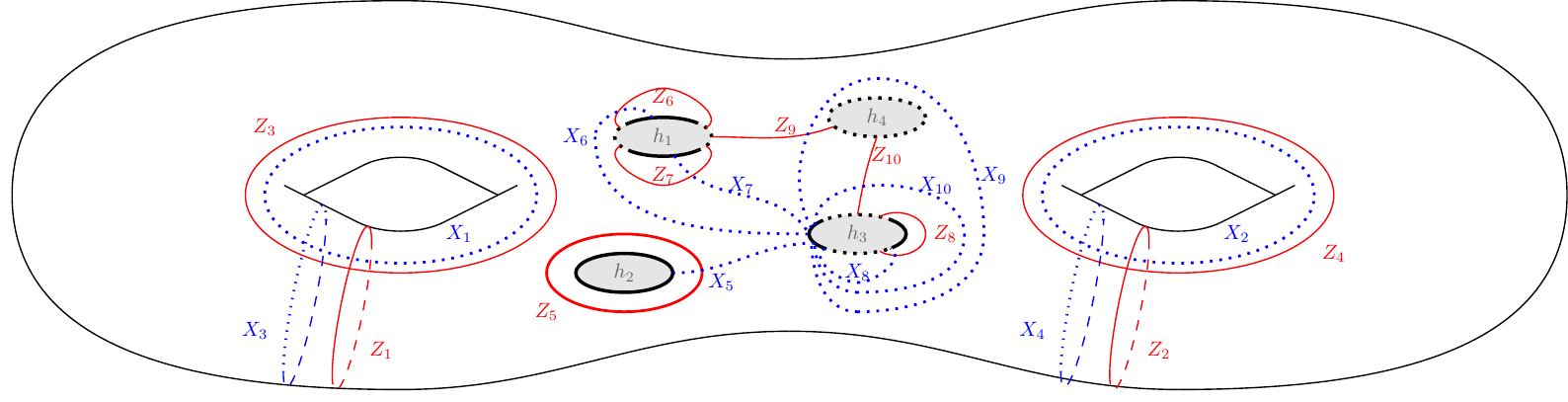}
\caption{A symplectic basis of the set of logical operators for an orientable surface of genus 2
punctured with holes with mixed boundaries.}
\label{fig:logical_mixed_punctured_2torus}
\end{center}
\end{figure} 

To conclude, let us propose a general strategy to construct a symplectic 
basis of logical operators for generalized surface codes. As before we can
attach 2 pairs of logical operators with any handle (or just one pair for 
non-orientable surface). The unusual set of logical operators comes from holes
and their boundaries.
The following strategy produces a symplectic basis to complete the handle 
logical operators in a symplectic basis.

Assume that the boundary of the surface contains $b$ holes and 
$m$ open paths. Denote by $h_1, \dots, h_{b_o}$ the $b_o$ holes
of the surface which contains at least one open boundary.
%Let $\gamma_1, \dots, \gamma_{\ell_o}$ the $\ell_o$ open paths (cyclic or not)
Let $\delta_1, \dots, \delta_{\ell_c}$ be the $\ell_c$ closed paths (cyclic or not)
along the boundary of the surface.
In what follows, for any path $\gamma$ of $G$ or its dual, 
$P(\gamma)$ denotes the Pauli operator which acts as the Pauli matrix 
$P = X$ or $Z$ on the qubits supported by $\gamma$.

\bigskip
{\bf Construction of $Z_1, \dots, Z_{b+m-2}$:}

1. For $i=1, \dots, \ell_c - 1$, define $Z_i = Z(\delta_{i})$.

2. For $i=1, \dots, b_o$, pick an open vertex $v_i$ on 
the boundary of the hole $h_i$.

3. For $i=1, \dots, b_o-1$, construct a path $\gamma_i$ that connects
$v_i$ to the last vertex $v_{\ell_o}$ and define $Z_{\ell_c-1+i} = Z(\gamma_i)$.

\medskip
{\bf Construction of $X_1, \dots, X_{b+m-2}$:}

1. For $i=1, \dots, \ell_c -1$ construct a path $\delta^*_i$ of the dual 
graph which connects $\delta_i$ to the last closed path $\delta_{\ell_c}$. 
and define $X_i = X(\delta^*_i)$.

2. For $i=1, \dots, b_o - 1$, construct a path $\gamma_i^*$ in the dual graph
that connects the path $\delta_{\ell_c}$ to itself after a loop around the hole 
$h_i$ and define $Z_i = Z(\gamma_{i}^*)$.

\medskip
This provides $\ell_c-1 + b_o-1$ independent logical operators of each type. 
One can easily see that $\ell_o-1 + \ell_c-1 = b+m-2$ is exactly the contribution 
of the holes and the open boundaries to $k$.
These operators satisfy the expected commutation relations of 
a symplectic basis.

\section{Homology of surfaces with mixed boundaries} \label{section:homology}

In this section, we develop a notion of homology appropriate to
surfaces with open and closed boundaries. Although homology of surfaces
with closed boundaries is well understood \cite{Ha02, Gi10}, we are not aware of 
any in-depth study of this peculiar notion of homology.
Recall that understanding homology of the underlying tiling is crucial for 
surface codes since it is responsible for  
the commutation relations of stabilizers and it governs the parameters of these quantum 
codes.

\subsection{Cycle space of open graphs}

As explained in Section~\ref{section:main_result}, cycles appear naturally
in the context of surface codes as the support of logical operators and
stabilizers. The number of encoded qubits of a generalized surface code will
be obtained by enumerating logical operators, that is cycles. This enumeration
relies on a closed formula for the number of cycles in a surface $G = (V, E, F)$,
obtained in the present section.

The notion of cycles and cocycles we derive in this section applies not only to a generalized surfaces $G = (V, E, F)$ with open and closed boundaries, but more generally to abstract graphs $H = (V,E)$ with open and closed boundaries.
Since open edges have no physical relevance (they do not support any qubit),
our goal is to compute the number of relative cycles of the graph $\mathring H = ( V, \mathring E)$.
First, we recall the definitions of a graph and a cycle and 
we extend these definitions to graphs with some open edges and vertices. 
We emphasize that open edges are mostly irrelevant throughout this section 
since we are concerned with the cycle structure of $\mathring H = ( V, \mathring E)$; their sole purpose is to define the notion of open vertices.

\medskip
A {\em graph} is defined to be a pair $H = (V, E)$, where $V$ is a finite set and 
$E$ is a set of pairs $\{u, v\}$ of elements of $V$. The elements of $V$ are called 
vertices and those of $E$ are the edges of the graph $H$. We assume that all 
graphs are finite and simple, {\em i.e.} have no loops nor multiple edges.
We refer to \cite{Be73} or \cite{Bo12} for standard results in graph theory.

\medskip
We consider a more general class of graphs that we call {\em open graphs},
where some edges are declared to be open. The set of open edges of 
$H$ is a subset of $E$ that we denote $\partial_{O} E$. The set of open vertices $\partial_O V$
is a subset of $V$ formed of the vertices that are adjacent to an open edge.  
We denote the interiors $\mathring V = V\backslash \partial_O V$ and $\mathring E = E\backslash \partial_O E$.
With these definitions, it is clear that the edges of $\mathring E$ have at most one endpoint in $\partial_O V$.
A graph, as defined in the previous paragraph, is simply an open graph 
with trivial open-edge set $\partial_O E = \emptyset$.

\medskip
In this paper, $\kappa(H)$ denotes the number of connected components of
the graph $H$. We will also have to enumerate the connected 
components of a graph $H$ satisfying some properties. For this purpose, we 
introduce the notation $\kappa_{X}(H)$ (respectively 
$\kappa_{\bar X}(H)$) for the number of connected components of the 
graph $H$ containing at least one element of the set $X$ (respectively no 
elements of the set $X$). The set $X$ will typically be a subset of the vertex set
or the edge set of $H$.

\medskip
A {\em cycle} in a graph $H=(V, E)$ is defined to be a subset $\gamma \subset E$ of 
edges of the graph $H$ that meets every vertex an even number of times. For an open graph, 
we only require that $\gamma$ meets the vertices of 
$V \backslash \partial_O V$ an even number of times. These cycles were called relative
cycles in the statement of Theorem~\ref{theo:code_parameters} and will simply be
called cycles or cycles of an open graph in what follows.
For instance, a path connecting two different open vertices is a cycle. 
If $\partial_O V$ is empty, these definitions coincide.

\medskip
The structure of the set of cycles of a graph, which we denote ${\cal C}(H)$ is well
understood when $H$ has no open vertices. It is a $\F_2$-linear space, where the 
sum of two cycles is defined as their symmetric difference.
This set is called the {\em cycle space} of $H$ and its dimension is
\begin{equation} \label{eq:cycle_code}
\dim {\cal C}(H) = |E| - |V| + \kappa(H).
\end{equation}
See for instance \cite{Be73} for an elegant proof of this result.

Here, we prove that the set of cycles of an open graph has a quite similar 
structure.
\begin{prop} \label{prop:cycle_code_boundaries}
The set of cycles of the interior $\mathring H$ of an open graph $H=(V, E)$ with open-vertex set
$\partial_O V$ and open-edge set
$\partial_O E$ is a $\F_2$-linear code of dimension
$$
\dim {\cal C}(\mathring H) = |\mathring E| - |\mathring{V}| + \kappa_{\overline{\partial_O V}}(\mathring H).
$$
\end{prop}
Recall that $\kappa_{\overline{\partial_O V}}(\mathring H)$ denotes the number of connected 
components of $\mathring H$ containing no open vertex $v \in \partial_O V$.
As expected we recover Eq.\eqref{eq:cycle_code} for graphs with an empty
open-vertex set.

\begin{proof}
To prove that the set of cycles of $\mathring H$ is a $\F_2$-linear code, it suffices to check 
that the sum (that is the symmetric difference) of two cycles is also a cycle.

The only non-trivial point is the dimension formula. 
In order to prove it, we construct a 2-cover $H_2$ of the graph $\mathring  H$ which has 
no open vertex. Eq.\eqref{eq:cycle_code} provides the dimension of the cycle 
code of $H_2$. Then, we derive the dimension of the cycle space of $\mathring H$ from
the one of $H_2$.
This idea of introducing a double cover to simplify our problem is quite common. 
For instance, it is standard to replace a non-orientable surface by its orientable double cover.

This 2-cover is constructed by taking two copies of the graph $\mathring H$ 
and by connecting the corresponding pair of open vertices.
Formally, the vertex set of $H_2 = (V_2, E_2)$ is $V_2 = V \times \F_2$.
Two vertices $(u, x)$ and $(v,y)$ in $V_2$ are linked by an 
edge if and only if $x=y$ and $\{u, v\}$ is an edge of $\mathring H$.
Moreover, an extra edge is added between all the pairs of vertices $(v, 0)$ 
and $(v, 1)$ associated with an open vertex $v \in \partial_O V$.

The graph $\mathring H$ is embedded in $H_2$ as the subgraph induced by the vertices
of the form $(u, 0)$.
This restriction that maps $H_2$ to $\mathring H$ induces a projection of the cycle space of 
$H_2$ onto the cycle space of $\mathring H$
\begin{align*}
\pi: {\cal C}(H_2) & \longrightarrow {\cal C}(\mathring H)
\end{align*}
which sends a cycle $\gamma \subset E_2$ onto its restriction $\gamma_0$ to the 
set of edges of the form $\{(u, 0), (v, 0)\}$. 
This linear application $\pi$ is surjective since any cycle 
$\gamma$ of $\mathring H$ is the image of the cycle of $H_2$ obtained from two copies 
$\gamma \times \{0\}$ and $\gamma \times \{1\}$ of the cycle $\gamma$,
and connecting them through the open vertices incident to $\gamma$.

From Eq.\eqref{eq:cycle_code}, the cycle space of $H_2$ has dimension
$$
\dim {\cal C}(H_2) = |E_2| - |V_2| + \kappa(H_2) 
= 2|E| + |\partial_O V| - 2|V| + 2\kappa_{\overline{\partial_O V}}(\mathring H) + \kappa_{\partial_O V}(\mathring H),
$$
where $\kappa_{\partial_O V}(\mathring H)$ is the number of connected components of $\mathring H$ 
containing at least one open vertex and $\kappa_{\overline{\partial_O V}}( \mathring H)$ denotes 
the number of connected components of $\mathring H$ containing no open vertices. 

Consider the kernel of the projection $\pi$.
A cycle of $H_2$ has a trivial image under $\pi$ if and only if it is a cycle 
of the subgraph $\bar H$ of $H_2$ induced by the vertices $(u, 1)$ 
where no vertices are declared to be open. We can therefore apply 
Eq.\eqref{eq:cycle_code}, which proves that
$$
\dim \Ker \pi = \dim {\cal C} (\bar H) = |E| - |V| + \kappa(\mathring H).
$$

By the rank nullity theorem, the dimension of the cycle space of $\mathring H$ is thus
$$
\dim {\cal C}(H_2) - \dim \ker \pi 
= |E| - |V| +  |\partial_O V| + 2 \kappa_{\overline{\partial_O V}}(\mathring H) + \kappa_{\partial_O V}(\mathring H) - \kappa(\mathring H).
$$
To conclude, note that $\kappa(\mathring H) = \kappa_{\overline{\partial_O V}}(\mathring H) + \kappa_{\partial_O V}(\mathring H)$.
\end{proof}

\subsection{Definition of Homology}

Intuitively, the first homology group of a tiling $G=(V, E, F)$ represents 
the cycles of the graph up to deformations. To formally define this group, 
we associate the following {\em chain complex} with the surface $G = (V, E, F)$.
$$
C_2 \overset{\partial_2}{\longrightarrow} C_1 \overset{\partial_1}{\longrightarrow} C_0.
$$
In this notation, $C_2, C_1$ and $C_0$ are the 3 $\F_2$-linear spaces
$$
C_0 = \bigoplus_{v \in \mathring{V}} \F_2 v, \quad C_1 = \bigoplus_{e \in \mathring{E}} \F_2 e, \quad C_2 = \bigoplus_{f \in F} \F_2 f.
$$
This triple is equipped with two $\F_2$-linear maps
$
\partial_2: C_2 \rightarrow C_1
$
and
$
\partial_1: C_1 \rightarrow C_0
$
defined by 
$$
\partial_2(f) = \sum_{\substack{e \in \mathring{E} \\ e \in f}} e
\quad \text{ and } \quad 
\partial_1(e) = \sum_{\substack{v \in \mathring{V} \\ v \in e}} v.
$$
These maps are called \emph{boundary maps}.
This definition is motivated by the fact that for the surface codes introduced 
in Definition~\ref{defi:surface_codes}, a $Z$-error $E_Z \in \{I, Z\}^{\otimes n}$ 
can be detected through the measurements of the vertex operators $X_v$ corresponding 
to non-open vertices $v \in \mathring V$.

\medskip
Before going to the definition of homology groups, recall that any subset $S \subset X$
corresponds to a vector $\delta_S$ of the space $\oplus_{x \in X} \F_2 x$ by the map 
$S \mapsto \delta_S = \sum_{x \in S} x$. This allows us to consider subsets of vertices, edges 
or faces respectively as vectors of $C_0$, $C_1$ or $C_2$ and to put a geometric 
intuition on these sets. For instance, a subset of $E$ corresponds to a vector of $C_1$.
A subset of edges corresponds to a vector of the kernel of the application $\partial_1$
if and only if it is called a cycle of the graph $(V, \mathring{E})$ where the vertices of the 
subset $\partial_O V$ are declared to be open.

\medskip
The following lemma motivates the choice of the chain spaces. It is also the key 
ingredient making the definition of surface codes possible on any surfaces
with or without boundaries. More precisely, it will be used to justify the 
commutation relations between the stabilizer generators of a surface code.
\begin{lemma} \label{lemma:boundary_composition}
For any tiling $G$, with or without boundaries, the composition of the 
two boundary maps is trivial: $\partial_1 \circ \partial_2 = 0$.
\end{lemma}

\begin{proof}
By linearity, it suffices to check that the image of a face $f \in F$ under
$\partial_1 \circ \partial_2$  is trivial.
A face of $G$ is an elementary cycle of the form 
$f = \{ \{v_1, v_2\}, \{v_2, v_3\} \dots, \{v_{\ell}, v_1\} \}$. 
If it contains only non-open edges (and by consequence only 
non-open vertices), then clearly $\partial_1 \circ \partial_2$
vanishes.
Otherwise, we assume that only the last edge $\{v_{\ell}, v_1\}$
is open. Its two endpoints are thus open vertices.
Then, we get
\begin{align*}
\partial_1 \circ \partial_2 (f) 
& = \partial_1( \sum_{i=1}^{\ell-1} \{v_i, v_{i+1}\} )\\
& = \partial_1(\{v_1, v_2\}) +  \partial_1( \sum_{i=2}^{\ell-2} \{v_i, v_{i+1}\} ) + \partial_1(\{v_{\ell -1}, v_{\ell}\})\\
& = v_2 + \sum_{i=2}^{\ell-2} (v_{i} + v_{i+1}) + v_{\ell-1} = 0.
\end{align*}
Adapting this argument to prove the general case where more edges 
of $f$ are open is straightforward.
\end{proof}

It follows directly from the previous lemma that $\im \partial_2 \subset \Ker \partial_1$.
This allows us to define the first homology group of a surface.
\begin{defi} \label{defi:H_1}
The \emph{first homology group} of a tiling $G$ with open and closed boundaries, denoted 
$H_1^\partial(G)$, is defined to be the quotient space
$
H_1^\partial(G) = \Ker \partial_1 / \im \partial_2.
$
\end{defi}
In the present work, we consider $\F_2$-homology. As we see in the previous
definition, the group $H_1^\partial(G)$ is then a $\F_2$-linear space.
The symbol $\partial$ in $H_1^\partial(G)$ is used to distinguish our generalization 
of the first homology group to the standard homology group $H_1(G)$ usually considered
for surfaces without boundaries.

\subsection{Rank of the first homology group}

Let us determine the dimension of $H_1^\partial(G)$. Recall that 
$\kappa_{\overline{\partial_C E}}(G)$ denotes the number of connected 
components of $G$ containing no closed boundary $e \in \partial_C E$ and 
$\kappa_{\overline{\partial_O V}}(G)$ is the number of connected components 
of $G$ containing no open vertex $v \in \partial_O V$.

\begin{prop} \label{prop:dim_H_1}
The dimension of $H_1^\partial(G)$ is given by
$$
\dim H_1^\partial(G) = - |\mathring{V}| + |\mathring{E}| - |F| + \kappa_{\overline{\partial_O V}}(G)  + \kappa_{\overline{\partial_C E}}(G),
$$
\end{prop}

For instance, for a connected orientable surface $G_{g, b}$ of genus $g$ 
with $b$ closed boundaries and no open boundaries, we recover the 
classical formula
$$
\dim H_1^\partial(G_{g, 0}) = 2g \quad \text{ and } \quad \dim H_1^\partial(G_{g, b>0}) = 2g+b-1.
$$
%b= 0: \dim H_1(G_{g, 0}) = - |V|+|E|- |F|+2\kappa(G) = 2g.
%b>0: \dim H_1(G) = - |V|+|E|- |F|+\kappa(G)+\kappa_{\bar C}(G) = 2g + b- 1
Therein, we used the following property of Euler-Poincar\'e characteristic:
$|V| - |E| + |F| = 2 - 2g - b$.

\begin{proof}
In order to compute the dimension of this quotient space, it suffices to
determine the dimension of $\Ker \partial_1$ and $\im \partial_2$.

The space $\Ker \partial_1$ corresponds to the cycle space of the 
open graph $(V, \mathring{E})$ with open 
vertex set $\partial_O V$. Its dimension, provided by 
Proposition~\ref{prop:cycle_code_boundaries}, is
$$
\dim \Ker \partial_1 = |\mathring{E}| - |\mathring{V}| + \kappa_{\overline{\partial_O V}}(G).
$$

The space $\im \partial_2$ is generated by the vectors $\partial_2(f)$
for $f\in F$. 
Consider a connected component $C$ of the tiling $G$.
Clearly, if the induced subtiling $C = (V_C, E_C, F_C)$ contains a closed boundary 
$e \in \partial_C E$ then the only relation of the form 
$\sum_{f \in F_C} \lambda_f \partial_2(f) = 0$ is the trivial one.
That means that the vectors $\partial_2(f)$ associated with the faces 
of this component are independent.
Assume now that $C$ contains no closed boundary.
Then each edge $e \in E \backslash \partial_O E$ belongs to exactly 0 or 2 
faces of $C$. This implies the non-trivial relation 
$\sum_{f \in F_C} \partial_2(f) = 0$,  proving that the vectors 
$\partial_2(f)$ are not independent. However, by a similar argument 
any $|F_C|-1$ of these faces are independent. Considering all the 
connected components of $G$ together, this proves that the dimension 
of $\im \partial_2$ is given by
$$
\dim \im \partial_2 = |F| - \kappa_{\overline{\partial_C E}}(G).
$$
where $\kappa_{\overline{\partial_C E}}(G)$ denotes the number of connected 
components of $G$ having no closed boundary $e \in \partial_C E$.

Altogether, we obtain the dimension of $H_1^\partial(G)$ by
$
\dim H_1^\partial(G) = \dim \Ker \partial_1 - \im \partial_2.
$
\end{proof}

\subsection{Local structure of the dual} \label{section:local_dual}

In order to prepare the construction of the dual of a combinatorial surface 
with boundaries, we provide in this section a precise description of the local
structure of a surface around a vertex. We will construct a dual face 
from the set $F_v$ of faces indicent to a vertex $v$. Figure~\ref{fig:local_dual}
illustrates this construction. 
We do not consider the distinction between an open and a closed boundary yet. 

By construction, the faces of any surface $(V, E, F)$ are glued together in such a way 
that
\begin{enumerate}
\item Any two faces meet in at most one edge.
\item Any edge belongs to either one or two faces.
\item Given any vertex $v \in V$, denote by $F_v$ the set of faces incident 
to $v$ and consider the graph obtained by adding an edge between any 
two elements of $F_v$ if the corresponding faces of $G$ share an edge. 
Then, this graph $F_v$ is either a self-avoiding path or a cycle.
\end{enumerate}

We previously defined boundaries through the edges that are on the boundary of a
unique face. Alternatively, one may define boundary vertices from the third property above 
by saying that a vertex $v$ is a boundary if the corresponding set $F_v$ induces a 
self-avoiding path. Then we would define boundary edges and faces as those incident
to a boundary vertex. This leads to an equivalent definition since a vertex $v$ is a 
boundary if and only if the corresponding set $F_v$ induces a self-avoiding path. The 
following lemma summarizes these equivalent definitions of boundaries.
\begin{lemma} \label{lemma:boundaries}
Let $G = (V, E, F)$ be a tiling with boundaries.
\begin{itemize}
\item $v \in \partial V$ 
$\Leftrightarrow$ $F_v$ induces a self-avoiding path
$\Leftrightarrow$ $v$ is incident to 2 edges of $\partial E$.
\item $e \in \partial E$ 
$\Leftrightarrow$ $e$ is incident to a unique face 
$\Leftrightarrow$ the two endpoints of $e$ belong to $\partial V$.
\item $f \in \partial F$ 
$\Leftrightarrow$ $f$ is incident to a vertex of $\partial V$ 
$\Leftrightarrow$ $f$ is incident to an edge of $\partial E$.
\end{itemize}
\end{lemma}
This lemma is clear. It is stated only for convenience.

\begin{figure}[h]
\centering
\begin{tikzpicture}[scale=2]

\tikzstyle{every node}=[circle, draw, fill=black!100, inner sep=1pt, minimum width=2pt]
\draw
	(0,0) node (v0) {}
	(1,.1) node (v1) {}
	(.7,.8) node (v2) {}	
	(.3,1.1) node (v3) {}
	(-.2,1.15) node (v4) {}
	(-.6,.85) node (v5) {}
	(-1,.3) node (v6) {}
	(-1.05,-.3) node (v7) {}
	(-.65,-.7) node (v8) {}
	(-.3,-1) node (v9) {}
	(.2,-.95) node (v10) {}
	(.7,-.6) node (v11) {};	

\draw
	(v0) -- (v1)
	(v0) -- (v2)
	(v0) -- (v5)
	(v0) -- (v7)
	(v0) -- (v8)
	(v0) -- (v11)	
	(v1) -- (v2) -- (v3) -- (v4) -- (v5) -- (v6) -- (v7) -- (v8) -- (v9) -- (v10) -- (v11) -- (v1);
	
\fill[color=blue!20, opacity=.5]
	(v1.center) -- (v2.center) -- (v3.center) -- (v4.center) -- (v5.center) -- (v6.center) -- (v7.center) -- (v8.center) -- (v9.center) -- (v10.center) -- (v11.center) -- (v1.center);
	
\draw
	(.6, .3) node (d1) {}
	(.02, .7) node (d2) {}
	(-.65, .2) node (d3) {}
	(-.55, -.33) node (d4) {}
	(-.05, -.7) node (d5) {}
	(.6, -.2) node (d6) {}
	;
	
\draw[dashed]
	(d1) -- (d2) -- (d3) -- (d4) -- (d5) -- (d6) -- (d1);
	
%=================

\draw[xshift = 3cm]
	(0,0) node (v0) {}
	(1,-.1) node (v1) {}
	(.7,.8) node (v2) {}	
	(.3,1.1) node (v3) {}
	(-.2,1.15) node (v4) {}
	(-.6,.85) node (v5) {}
	(-1,.3) node (v6) {}
	(-1.05,-.5) node (v7) {};
	
\draw[xshift = 3cm]
	(v0) -- (v1)
	(v0) -- (v2)
	(v0) -- (v5)
	(v0) -- (v7)
	(v1) -- (v2) -- (v3) -- (v4) -- (v5) -- (v6) -- (v7);

\draw[xshift = 3cm, line width = 1.5pt]
	(v0) -- (v1)
	(v0) -- (v7);
	
\fill[xshift = 3cm, color=blue!20, opacity=.5]
	(v1.center) -- (v2.center) -- (v3.center) -- (v4.center) -- (v5.center) -- (v6.center) -- (v7.center) -- (v0.center) -- (v1.center);
	
\draw[xshift = 3cm]
	(.6, .3) node (d1) {}
	(.02, .7) node (d2) {}
	(-.65, .2) node (d3) {};
	
\draw[xshift = 3cm, dashed]
	(d1) -- (d2) -- (d3);
	
\draw
	(v0)--(v1) node [midway] (d4) {}
	(v0)--(v7) node [midway] (d5) {};
	
\draw[dashed]
	(d3) -- (d5)
	(d1) -- (d4)
	(d4) -- (d5);
	
\end{tikzpicture}
\caption{Local structure of the dual graph around a vertex $v$. 
The faces containing a vertex $v$ are represented.
If $v$ is not a boundary (left), then $F_v$ induces a cycle in the dual graph
represented by dashed lines.
If $v$ is a boundary (right), then $F_v$ is completed by adding 2 vertices 
in the middle of the 2 boundary edges and by connecting them by a dual edge.
This results in a cycle $\bar F_v$ in the dual graph. 
The two dark edges at the right are the boundaries.}
\label{fig:local_dual}
\end{figure}
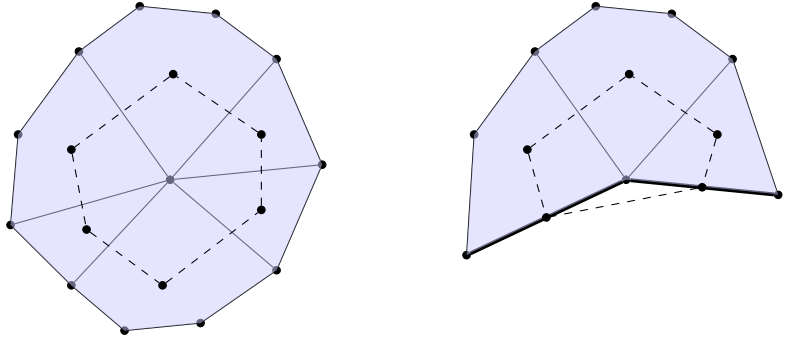

\medskip
We now show that the set $F_v$ introduced in item 3 above can always be endowed 
with a structure of cycles as depicted in Figure~\ref{fig:local_dual}.
It will later play the role of a face of the dual surface. 
Let $v \in V$ be a vertex of a surface $G$ and let $F_v$ be the set of faces of $G$
incident to $v$. Denote by $v_f$ the elements of $F_v$ where $f \in F$ is a face incident 
to $v \in V$.
Let $\bar F_v$ be the set $F_v$ completed with the 2 extra elements $v_e$ and $v_{e'}$ 
coming from the 2 boundary edges $e, e' \in E$ incident to $v$.
The set $\bar F_v$ is regarded as a vertex set and equipped with the edges $\{v_f, v_{f'}\}$ 
such that $f$ and $f'$ share an edge, the 2 edges $\{v_e, v_f\}$ where $f$ is the unique face 
containing $e$, and finally the edge $\{v_e, v_{e'}\}$ connecting the 2 boundaries.
When $v$ is not a boundary, $\bar F_v$ coincides with $F_v$.
\begin{lemma} \label{lemma:dual_faces}
Let $v \in V$ and let $F_v$ and $\bar F_v$ be the corresponding local dual graphs, then
the graph $F_v$ is a cycle if and only if $v \in V \backslash \partial V$ and 
$\bar F_v$ is always a cycle.
\end{lemma}
This result is an immediate consequence of the previous lemma.
The two possible configurations for the local dual graph are represented in Figure~\ref{fig:local_dual}.
The cycle $\bar F_v$ is not necessarily included in the original surface but it can 
always be deformed to be embedded in the surface. This makes it an ideal candidate to define 
dual faces.

\subsection{Dual cellulation} \label{section:dual}

The dual cellulation $G^*$ of a combinatorial surface $G$ without boundaries 
is simply obtained by replacing each face of $G$ by a vertex and by connecting
two such vertices if the corresponding faces of the original graph 
share an edge. This leads to a correspondence between the edges
of the graph and those of its dual. 
This dual graph naturally defines a cellulation of the same surface
whose faces are given by the cycle $F_v$ (see Lemma~\ref{lemma:dual_faces}).
The faces of $G^*$ are therefore in one-to-one correspondence with 
the vertices of $G$.

\medskip
Before stating our definition of the dual of a surface with open and closed
boundaries, let us consider two natural notions of dual for surface with
boundaries without considering the type (open or closed) of the boundaries.
\begin{itemize}
\item {\bf Inner Dual:} Following the standard construction, one could 
define a dual graph by replacing each face of $G$ by a vertex and connecting two 
vertices if they correspond to faces sharing an edge. Then by 
Lemma~\ref{lemma:dual_faces}, we see that all the vertices 
$v \in V \backslash \partial V$ induce a cycle $F_v$ in the dual graph. This 
defines the face set of the dual. No face is associated with boundary vertices.
This corresponds to cutting the regions on the boundary of the surface where 
faces are not trivially defined.
\end{itemize}
This construction preserves the bijection between the faces of the graph and the
vertices of its dual. However an edge of $G$ that belongs to a single face, has no 
corresponding edge in the dual. The correspondence between the vertices of $G$ and 
the faces of its dual is also lost.
We may also consider the following extension of the inner dual.
\begin{itemize}
\item {\bf Outer Dual:} 
Another dual can be obtained by replacing each vertex of $G$ by a face
of the dual. 
More precisely, from Lemma~\ref{lemma:dual_faces}, we can replace 
any vertex $v \in V$, boundary or not, by the face defined by the cycle $\bar F_v$.
Then these faces are glued together as follows. For any pair of neighbours vertices $u$
and $v$, the faces $\bar F_u$ and $\bar F_v$ are stuck together along the
edge corresponding to the edge $\{u, v\}$.
\end{itemize}
This latter dual re-establishes the correspondence between the vertices of $G$ and the faces of $G^*$
but it does not preserve the bijections between $E$ and $E^*$ or between $F$ and $V^*$.

\medskip
The definition of a dual surface is thus less straightforward in the presence of 
boundaries. For our purpose, we also need to consider the type of the boundaries.
In order to preserve the surface code structure, we aim for a notion of duality which
leads to a correspondence (i) between $\mathring E$ and $\mathring{E}^*$ (like in 
the enlarged dual) since $\mathring E$ supports the qubits, (ii) between $\mathring V$
and $F^*$, and (iii) between $F$ and $\mathring{V}^*$. The bijections (ii) and (iii) are 
required to exchange the roles of $X$ and $Z$ through duality.

\medskip
We will proceed in two steps. Roughly, we define the dual surface as the inner dual 
along open boundaries and as the outer dual along closed boundaries.
Through steps 1 to 3 below, we define a dual graph for which the correspondence (i) 
is restored. Then, we extend this graph by adding some extra edges (the open edges 
in step 4) to properly define faces. This makes it a combinatorial surface and this
brings back the correspondence (ii) and (iii). We can check that the correspondences
are satisfied in Table~\ref{tab:correspondences_duality}.

\medskip
The dual surface $G^*=(V^*, E^*, F^*)$ is obtained from $G=(V, E, F)$ 
by the following process.
\begin{enumerate}
\item {\bf Non-open vertices}: Define a vertex $v_f \in V^*$ for each face $f \in F$

\item {\bf Open vertices}: Define an open vertex $v_e \in \partial_O V^*$ for 
each closed edge $e \in \partial_C E$ 

\item {\bf Non-open edges}: 
For each edge $e \in E \backslash \partial E$, add an edge $\{v_f, v_{f'}\} \in E^*$,
where $f$ and $f'$ are the two distinct faces containing $e$.
For each edge $e \in \partial_C E$, add an edge $\{v_e, v_f\} \in E^*$,
where $f \in F$ is the unique face of $G$ containing the boundary edge $e$.

\item {\bf Open edges}: Add an edge $\{v_e, v_{e'}\} \in \partial_O E^*$ for 
every pair of distinct edges $e, e' \in \partial_C E$ sharing a vertex $v \in \partial_C V$ in $G$.

\item {\bf Non-open faces}: 
Define a face from the cycle $F_v$ for each $v \in V \backslash \partial V$.

\item {\bf Open faces}: 
Define a face from the cycle $\bar F_v \in \partial_O F^*$ for each $v \in \partial_C V$.
\end{enumerate}

\begin{table}[hb]
\centering
\footnotesize
\caption{Correspondences between a surface and its dual.}
\label{tab:correspondences_duality}
%\centerline{\footnotesize\smalllineskip
\begin{tabular}{c c}
Tiling $(V, E, F)$ & Dual tiling $(V^*, E^*, F^*)$\\
\hline
$F$ & $\mathring{V}^*$\\
$\partial_C E$ & $\partial_O V^*$\\
\hline
$\mathring{E}$ & $\mathring{E}^*$\\
$\partial_C V$ & $\partial_O E^*$\\
\hline
$V \backslash \partial V$ & $\mathring{F}^*$\\
$\partial_C V$ & $\partial_O F^*$\\
\hline
\end{tabular}
\end{table}

\begin{figure}[h]
\centering
\begin{tikzpicture}[scale=1]

\draw[yshift = -.5cm] (2,.5) node {(a)};
\draw[xshift = 5cm, yshift = -.5cm] (2.5,.5) node {(b)};
\draw[xshift = 10cm, yshift = -.5cm] (2.5,.5) node {(c)};

\tikzstyle{every node}=[circle, draw, fill=black!100, inner sep=0pt, minimum width=2pt]
%=============
%the first tiling
%=============
%nodes
\draw
\foreach \x in {1,2,...,3}
\foreach \y in {1,2,...,3}
{ (\x,\y) node {} };

%grid
\draw 
(1,1)--(1,3)
(2,1)--(2,3)
(1,1)--(3,1)
(1,2)--(3,2)
(1,3)--(3,3);
\draw[style=dashed]
(3,1)--(3,3);

%color faces
\fill[color=blue!20, opacity=.5] 
	(1,1)--(3,1)--(3,3)--(1,3)--(1,1);

%=============
%its dual graph
%=============
%nodes
\draw[xshift = 5cm, yshift = -.5cm]
\foreach \x in {1,2,...,3}
\foreach \y in {2,3}
{ (\x,\y) node {} };

\draw[xshift = 5cm,  yshift = -.5cm]
\foreach \x in {2,3}
\foreach \y in {1,2,...,4}
{ (\x,\y) node {} };

%lines
\draw[xshift = 5cm,  yshift = -.5cm, step=1cm] 
(1,2)--(3,2)
(1,3)--(3,3)
(2,1)--(2,4)
(3,1)--(3,4);

%=============
%its dual tiling
%=============
%nodes
\draw[xshift = 10cm, yshift = -.5cm]
\foreach \x in {1,2,...,3}
\foreach \y in {2,3}
{ (\x,\y) node {} };

\draw[xshift = 10cm,  yshift = -.5cm]
\foreach \x in {2,3}
\foreach \y in {1,2,...,4}
{ (\x,\y) node {} };

%edges
\draw[xshift = 10cm,  yshift = -.5cm, step=1cm] 
(1,2)--(3,2)
(1,3)--(3,3)
(2,1)--(2,4)
(3,1)--(3,4);

%open edges
\draw[xshift = 10cm,  yshift = -.5cm, step=1cm, style=dashed]
(1,2)--(2,1)--(3,1)
(3,4)--(2,4)--(1,3)--(1,2);

%color faces
\fill[xshift = 10cm,  yshift = -.5cm, color=blue!20, opacity=.5] 
	(3,4)--(2,4)--(1,3)--(1,2)--(2,1)--(3,1)--(3,4);

\end{tikzpicture}
\caption{(a) A surface with open boundaries represented by dashed lines. 
(b) Its dual graph is obtained (steps 1 to 3) by replacing each non-open 
edge by a non-open edge. 
(c) Faces of the dual surface are obtained by adding open edges to equip 
the dual graph with a structure of surfaces.}
\end{figure}
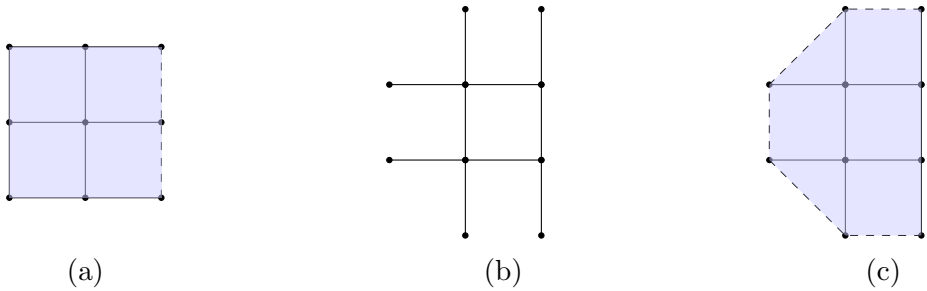

The different steps are chosen to emphasize the different correspondences
between the surface and its dual (See Table~\ref{tab:correspondences_duality}).
As previously, the non-open subsets of vertices, edges and faces of $G^*$ 
are denoted respectively by 
$\mathring{V}^* = V^* \backslash \partial_O V^*$, 
$\mathring{E}^* = E^* \backslash \partial_O E^*$
and 
$\mathring{F}^* = F^* \backslash \partial_O F^*$.
By construction, $G^*$ defines a cellulation of the same surface as $G$.
We can easily check that open vertices, edges and faces are defined in a 
coherent way.
This is a subgraph of the outer dual which contains the inner dual.

\medskip
{\bf Remark:} Property 2 of Section~\ref{section:local_dual} may not be satisfied by the dual
graph. If the 2 endpoints of a non-open edge are open then the dual of this edge will not belong
to any face of the dual as in Figure~\ref{fig:bad_dual}. 
In the error correction terminology, the surface code would have minimum distance 1. We can simply avoid 
this configuration by declaring any edge whose both endpoints are open as an open egde.

\medskip
{\bf Remark:} To avoid the presence of loop or multiple edges in the dual graph, one
can restrict the girth of $G$ to be at least 3.

\begin{figure}
\centering
\includegraphics[scale=.3]{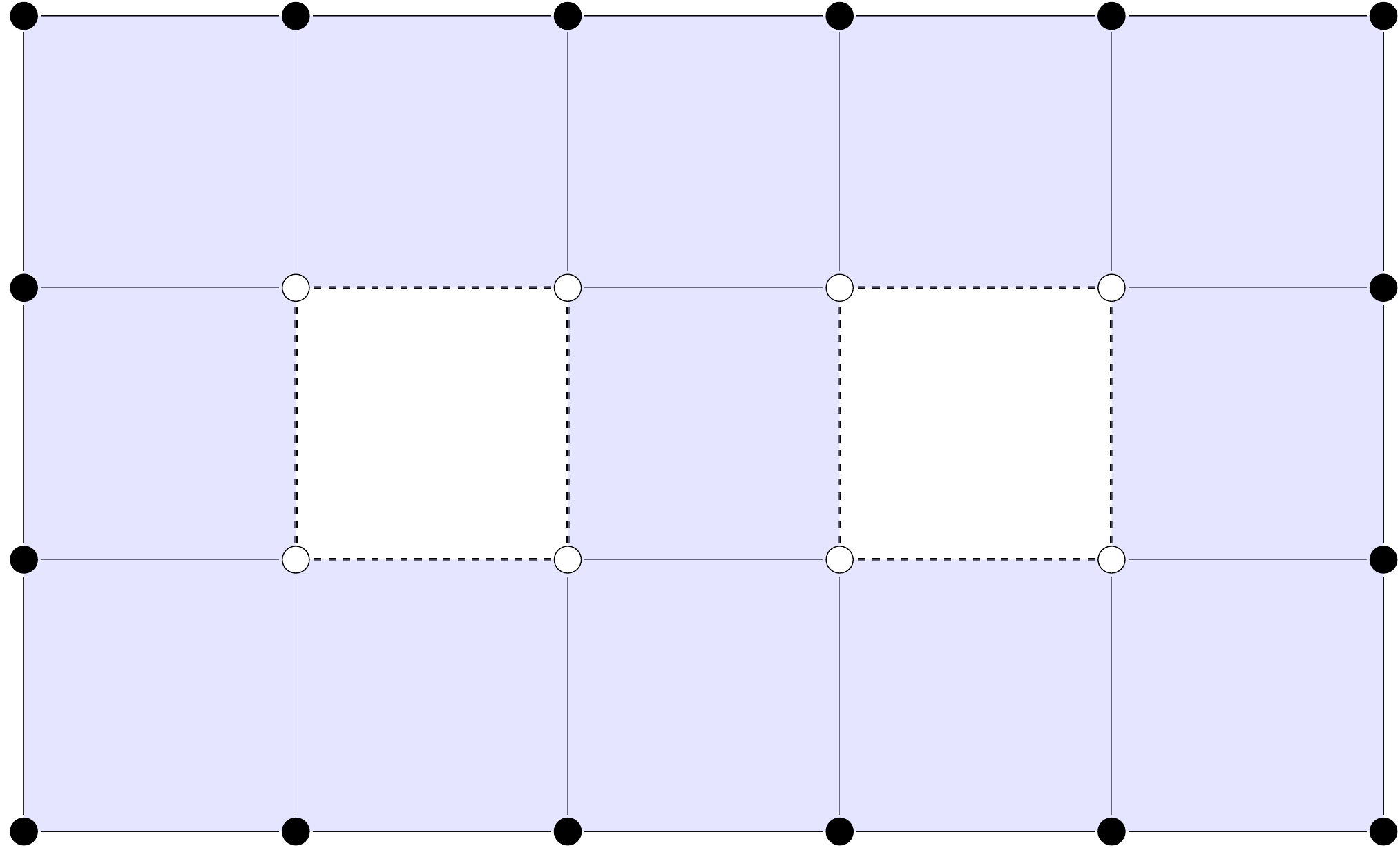}
\includegraphics[scale=.3]{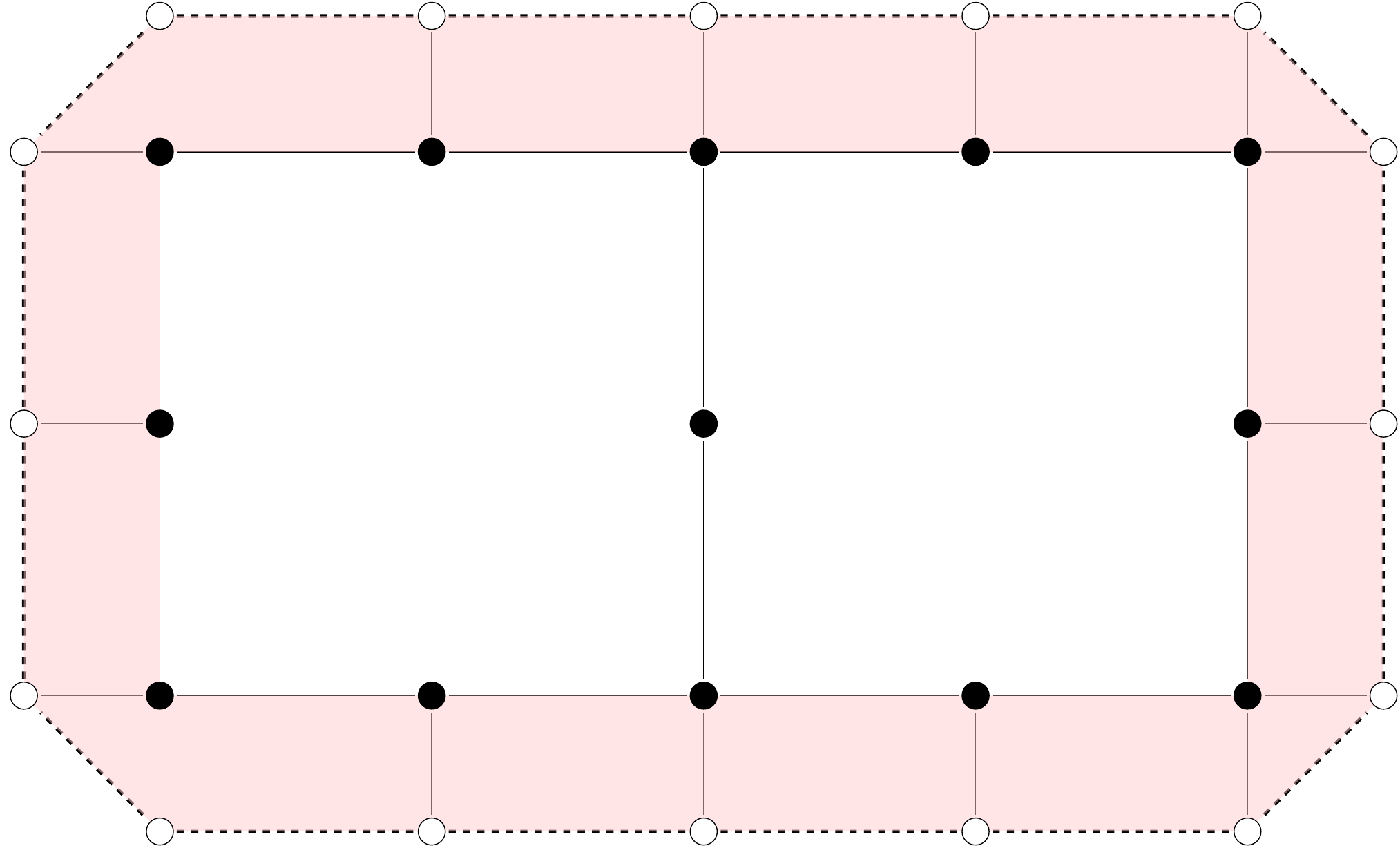}
\caption{Left: A tiling with open and closed boundaries. Dashed edges represent open edges and white vertices are open vertices. The tiling contains two non-open edges whose endpoints are both open vertices. Right: The two corresponding edges in the dual belong to no face. Then this dual tiling is not tiling.}
\label{fig:bad_dual}
\end{figure}

\subsection{Application to generalized surface codes}
\label{sec:GSC}

This notion of homology is chosen to correspond to errors on the surface of 
Def.~\ref{defi:surface_codes} and their syndromes. Indeed, by definition 
a $Z$-error $E_Z \in \{I, Z\}^{\otimes n}$ corresponds to a vector 
$z$ of $C_1$ defined by
$$
z = \sum_{e \in \Supp(E_Z)} e.
$$
Its syndrome $\sigma(E_Z)$ corresponds to the vector $\partial_1(z) \in C_0$ and 
$Z$-stabilizers are given by the vectors $z \in \im \partial_2$. 

Similarly, an error $E_X \in \{I, X\}^{\otimes n}$ corresponds to the vector 
$x = \sum_{e \in \Supp(E_X)} e$ of $C_1$. Two Pauli operators 
$E_X$ and $E_Z$ commute if and only if the corresponding vectors 
$x$ and $z$ are orthogonal for the binary inner product in $C_1$.

Through these isomorphisms, the set of $Z$-stabilizers can be seen as the image 
of the map $\partial_2$ and the $X$-stabilizers are in bijection with the vectors 
of $(\ker \partial_1)^\perp$.
From Lemma~\ref{lemma:boundary_composition}, we have
$\im \partial_2 \subset \ker \partial_1 = ((\ker \partial_1)^\perp)^\perp$. 
This proves that the two spaces corresponding to $S_X$ and $S_Z$ are orthogonal.
In other words, we just proved that the commutation relation between the operators 
$X_v$ and $Z_f$ is ensured by the relation $\partial_1 \circ \partial_2 = 0$:
\begin{coro}
For any surface $G$, the operators $X_v$ for $v \in \mathring{V}$ and $Z_f$ 
for $f \in F$ are commuting.
\end{coro}

We know that the number of logical qubits $k$ encoded in the surface $G$ is 
given by the rank of the quotient of the group of $Z$-errors of syndrome 0 
by the subgroup of $Z$-stabilizers. As an immediate application, $k$ is the dimension of the first 
homology group $H_1^\partial(G)$. Proposition~\ref{prop:dim_H_1} yields:
\begin{coro} \label{coro:k}
The number of logical qubits $k$ encoded using the surface code associated with $G$
is 
$$
k = - |\mathring{V}| + |\mathring{E}| - |F| + \kappa_{\overline{\partial_O V}}(G)  + \kappa_{\overline{\partial_C E}}(G).
$$
\end{coro}

Consider the minimum distance $d = \min(d_X, d_Z)$ of the code. We already noticed that 
a $Z$-error has trivial syndrome if and only if its support $z \in C_1$ is a cycle and it is not a
stabilizer if and only if this cycle has non-trivial homology. In order to obtain a graphical expression for
the minimum distance, we need a similar result for $d_X$. An error $E_X$ corresponds to a vector 
$x \in C_1$ but its syndrome and $X$-stabilizers are less trivial to describe. Consider this same error
in the dual graph $G^* = (V^*, E^*, F^*)$. Denote by
$$
C_2^* \overset{\partial_2^*}{\longrightarrow} C_1^* \overset{\partial_1^*}{\longrightarrow} C_0^*,
$$
the homology chain-complex associated with the dual surface $G^*$.
Using the correspondence of Table~\ref{tab:correspondences_duality}, we see that 
this duality map $G \rightarrow G^*$ transforms the vector $x = \sum_{e \in \Supp(E_X)} e \in C_1$
in the vector $x^* = \sum_{e \in \Supp(E_X)} e^* \in C_1$, where $e^*$ is the dual edge 
of $e$. Any stabilizer $Z_f$ is sent onto $Z_{v^*_f}$ where $v^*_f$ is the dual vertex 
associated with the face $f$. This relies on the bijection between faces of $G$ and non-open 
vertices of its dual. The syndrome of $E_X$ can therefore be expressed as $\partial_1^*(x^*)$.
The $X$-stabilizers $X_v$, which correspond to non-open vertices are mapped onto the 
vectors of $\im \partial_2^*$. This is also a consequence of the definition of the dual graph and the
correspondences of Table~\ref{tab:correspondences_duality}.
This proves that $d_X$ is the shortest length of a non-trivial cycle of $G^*$.
The next corollary follows.

\begin{coro} \label{coro:d}
$d_Z$ is the minimum length of a relative cycle of $G$ with non-trivial homology
in $H_1^\partial(G)$
and $d_X$ is the minimum length of a relative cycle of $G^*$ with non-trivial homology
in $H_1^\partial(G^*)$.
%The minimum distance $d$ of the surface codes associated with $G$ is 
%the minimum length of a cycle of $G$ or $G^*$ with non-trivial homology in 
%$H_1^\partial(G)$ or in $H_1^\partial(G^*)$.
\end{coro}

An equivalent proof of the graphical expression of $d_X$ might be obtained 
by considering the cellular cohomology of the surface $G$. An error 
$E_X$ of trivial syndrome which is not a stabilizer corresponds to a cocycle 
with non-trivial cohomology. Then it suffices to remark that the cohomology 
chain-complex of a surface is isomorphic to the homology complex of its dual.

\section{Packing of logical qubits in a planar architecture} \label{section:planar_codes}

In this section we argue that mixed holes with partially open and partially closed
boundaries may offer an advantage over usual surface codes for storage
of quantum information in a two-dimensional lattice.

Bravyi, Poulin and Terhal \cite{BPT10a} proved that the parameters of any two-dimensional local commuting
projector code over finite dimensional quantum systems embedded in a
square grid are subjected to the bound
$$
k d^2 \leq c n
$$
for some constant $c > 0$ that depends on the locality of the constraints.
This bound can be seen as a tradeoff between the amount of quantum 
information stored in the lattice (given by the number of logical qubits $k$) 
and the error-correction capability of the code (measured by the minimum 
distance $d$). It provides a natural figure of merit to compare different 
quantum computing architectures. Optimizing the constant $c$ can save 
a large amount of resources while keeping roughly the same performance.

\medskip
Let us first provide some intuition for the special case of generalized surface codes 
based on a planar lattice punctured with closed holes.
Each hole represents one logical qubit.
To preserve a large minimum distance $d$, these holes must be separated from
each other by a distance at least $d$. In other words, the neighborhoods 
$B(h, (d-1)/2)$ of each hole $h$ containing all the qubits within distance 
$(d-1)/2$ from hole $h$ do not overlap. 
These buffers of physical qubits surrounding each holes each consists $\Omega(d^2)$ qubits
due to the two-dimensional geometry of the lattice (here, we assume that our lattice is Eucliean,
for instance hyperbolic codes constructed in \cite{FML02, Ze09, DZ10, BT15:hyperbolic} are locally planar but 
are not subjected to this tradeoff \cite{De13:tradeoffs}).
This shows that encoding $k$ logical qubits using a uniform planar surface code with 
minimum distance $d$ requires at least $n = \Omega(kd^2)$ qubits. This argument also
emphasizes the resemblance with a sphere packing problem.

\medskip
In this section, we will use the formalism of generalized surface codes to improve the constant
$c$ over previously known constructions.

\subsection{Underlying lattice}

Let us fix some notations and definitions. We focus on a square lattice of
qubits. Starting with a finite region of a planar square lattice with closed 
boundaries, we will encode qubits as holes in this region. Two kinds of 
regions of the square lattice are considered in what follows, as depicted 
in Figure~\ref{fig:underlying_lattice}.

\begin{figure}[h]
\begin{center}
\includegraphics[scale=.2]{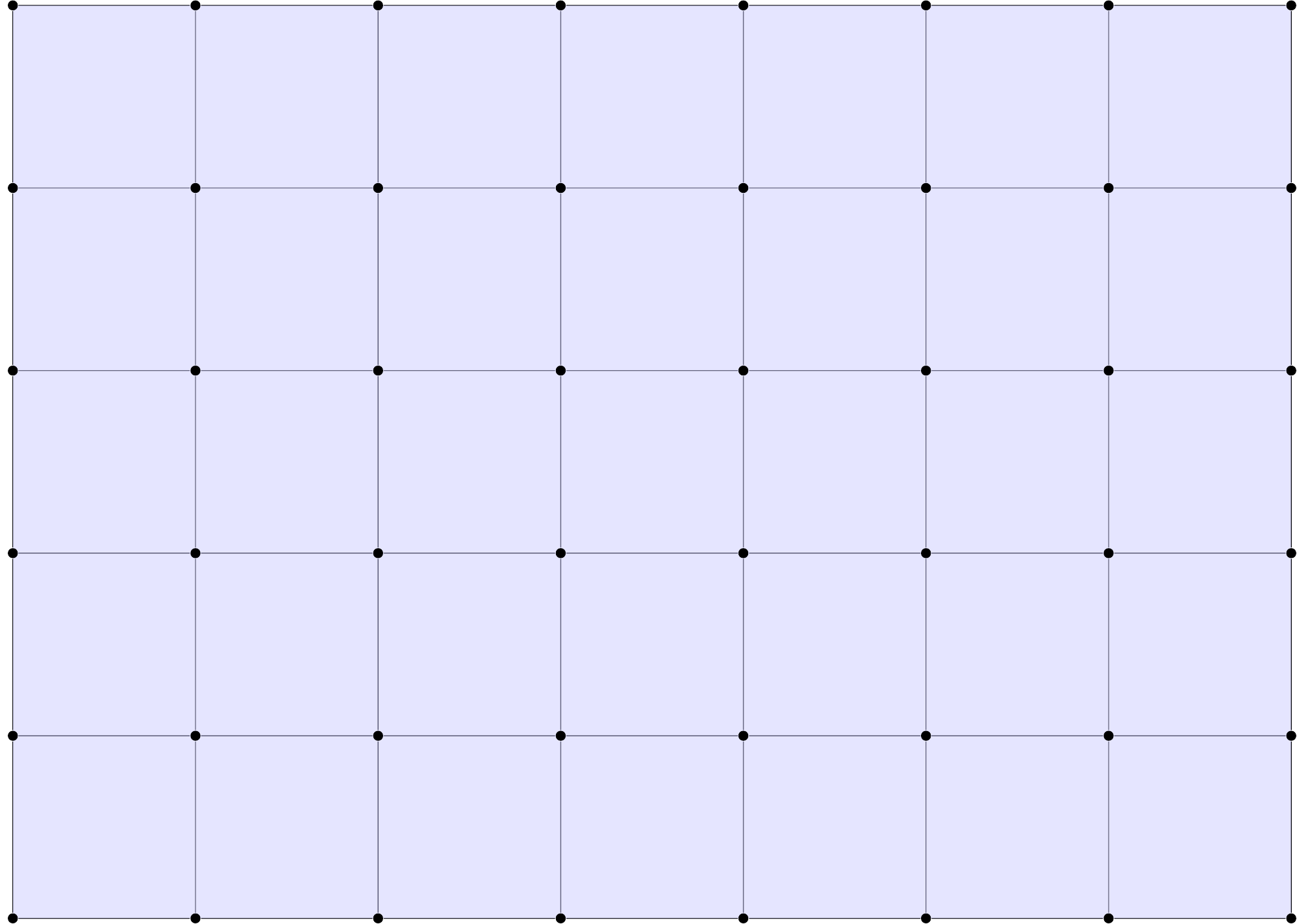}
\hspace{.2cm}
\includegraphics[scale=.2]{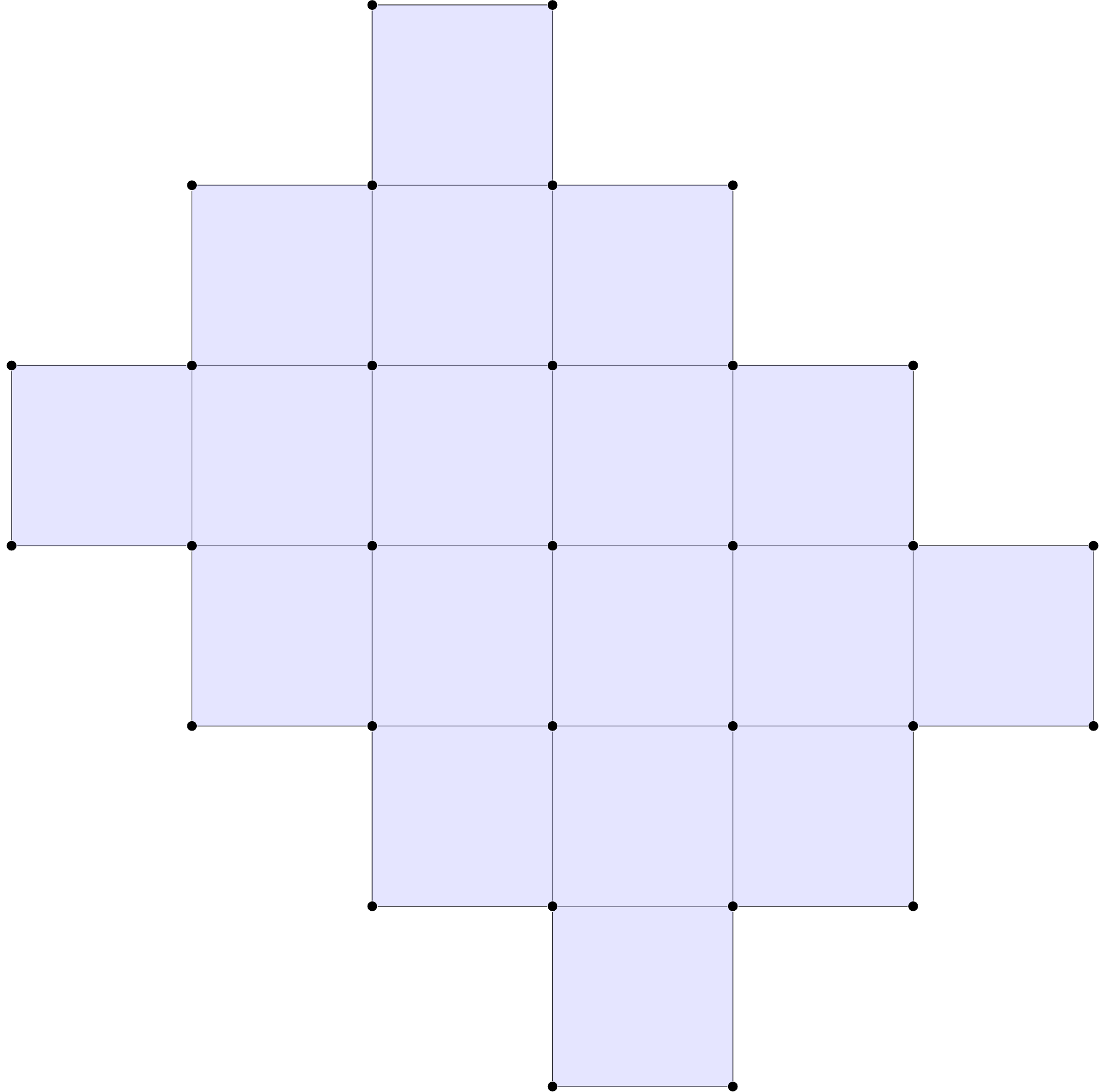}
\caption{Left: A $(7 \times 5)$ square lattice. Right: A $(3 \times 4)$-rotated square lattice.}
\label{fig:underlying_lattice}
\end{center}
\end{figure}

By a {\em $(L \times L')$-square lattice}, we mean the subgraph of $\Z^2$
induced by the vertices $(x, y) \in [0, L] \times [0, L']$.
This lattice contains 
%$(L+1)(L'+1) = LL' + L+L'+1$
$LL' + L+L'+1$ vertices, 
%$L (L'+1) + L'(L+1) = 2LL' + L+L'$ 
$2LL' + L+L'$ edges and 
$LL'$ faces.
See Figure~\ref{fig:underlying_lattice} for an example.

The second region of the square lattice that we use will be called
a {\em $(L \times L')$-rotated square lattice}. It is a patch of the square
lattice $\Z^2$ by a rectangle of length $\sqrt{2} L \times \sqrt 2 L'$ 
rotated by $45^\circ$. For instance, a $(3 \times 4)$-rotated square lattice 
is shown in Figure~\ref{fig:underlying_lattice}.
Such a rotated square lattice is said to have size $L \times L'$. 
It contains 
%$L(L'+1) + L'(L+1) = 2LL' + L + L'$
$2LL' + L + L'$ vertices, 
$4 L L'$ edges and 
%$LL' + (L-1)(L'-1) = 2LL' -L - L' +1$ 
$2LL' -L - L' +1$ faces.

\subsection{Square hole architecture}

As a first example, we consider a square lattice punctured with square holes. 
This construction which is one of the most widely studied quantum computing
architecture \cite{RH07, RHG06, RHG07, FMMC12}, will be referred to as the 
{\em square hole architecture} or the {\em square hole surface code} and 
will be denoted $\Sq(h, h', t)$ or simply $\Sq(h, t)$ when $h= h'$. 
The parameters $h$, $h'$ and $t$ are non-negative integers.
Each qubit corresponds to a $(t \times t)$ closed hole and these holes are separated to one another 
and to the boundary by a distance $4t-1$. This ensures that the minimum distance of the code 
is $d=4t$. This lattice is represented in Figure~\ref{fig:square_holes}.
The number of encoded qubits is given by the number of holes. In order to create 
$(h \times h')$ holes, this structure requires a $(L_h \times L_{h'})$-square lattice 
where $L_h = h(5t-1)+(4t-1)$. 

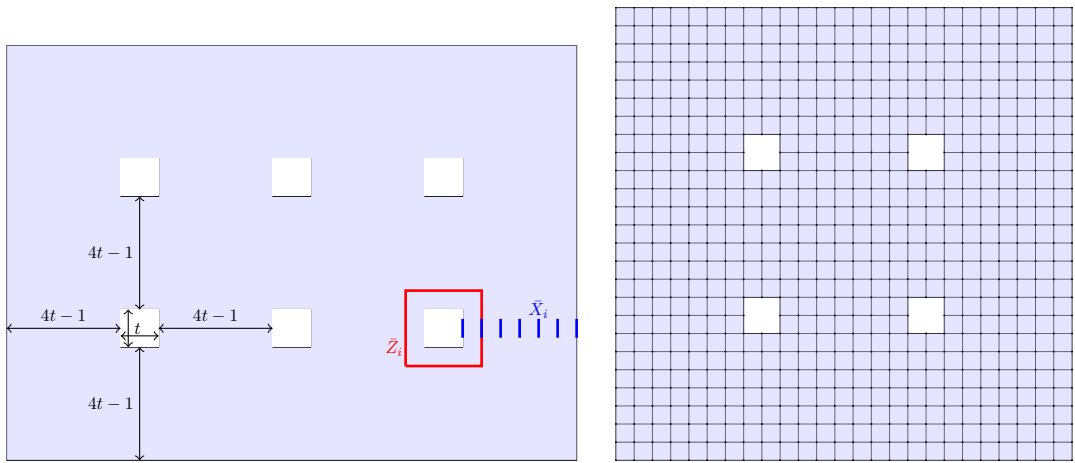
\begin{figure}[h]
\centering
\begin{tikzpicture}[scale=.5, every node/.style={scale=.6}]

%\tikzstyle{every node}=[circle, draw, fill=black!100, inner sep=1pt, minimum width=2pt]

%large rectangle
\draw
	(0,0) -- (15,0) -- (15,11) -- (0,11) -- (0,0);
\fill[color=blue!20, opacity=.5]
	(0,0) -- (15,0) -- (15,11) -- (0,11) -- (0,0);
	
%holes
\draw
	(3,3) -- (4,3) -- (4,4) -- (3,4) -- (3,3);
\fill[color=white]
	(3,3) -- (4,3) -- (4,4) -- (3,4) -- (3,3);

\draw[xshift = 4cm]
	(3,3) -- (4,3) -- (4,4) -- (3,4) -- (3,3);
\fill[color=white, xshift = 4cm]
	(3,3) -- (4,3) -- (4,4) -- (3,4) -- (3,3);

\draw[xshift = 8cm]
	(3,3) -- (4,3) -- (4,4) -- (3,4) -- (3,3);
\fill[color=white, xshift = 8cm]
	(3,3) -- (4,3) -- (4,4) -- (3,4) -- (3,3);

\draw[yshift = 4cm]
	(3,3) -- (4,3) -- (4,4) -- (3,4) -- (3,3);
\fill[color=white, yshift = 4cm]
	(3,3) -- (4,3) -- (4,4) -- (3,4) -- (3,3);

\draw[xshift = 4cm, yshift = 4cm]
	(3,3) -- (4,3) -- (4,4) -- (3,4) -- (3,3);
\fill[color=white, xshift = 4cm, yshift = 4cm]
	(3,3) -- (4,3) -- (4,4) -- (3,4) -- (3,3);

\draw[xshift = 8cm, yshift = 4cm]
	(3,3) -- (4,3) -- (4,4) -- (3,4) -- (3,3);
\fill[color=white, xshift = 8cm, yshift = 4cm]
	(3,3) -- (4,3) -- (4,4) -- (3,4) -- (3,3);
	
%arrows
\draw[<->]
	(3.2, 3) -- (3.2, 4) node [midway, right] {$t$};
\draw[<->]
	(3, 3.3) -- (4, 3.3);

\draw[<->]	
	(3.5, 3) -- (3.5, 0) node [midway, left] {$4t-1$};
\draw[<->, yshift=4cm]	
	(3.5, 3) -- (3.5, 0) node [midway, left] {$4t-1$};

\draw[<->]	
	(0, 3.5) -- (3, 3.5) node [midway, above] {$4t-1$};
\draw[<->, xshift=4cm]	
	(0, 3.5) -- (3, 3.5) node [midway, above] {$4t-1$};
	
%logical
\draw[color = red, line width = 1pt]
	(10.5,2.5) -- (12.5,2.5) -- (12.5,4.5) -- (10.5,4.5) -- (10.5,2.5);	

\draw[color = blue, line width = 1pt]
	(12,3.25) -- (12,3.75)
	(12.5,3.25) -- (12.5,3.75)
	(13,3.25) -- (13,3.75)	
	(13.5,3.25) -- (13.5,3.75)
	(14,3.25) -- (14,3.75)	
	(14.5,3.25) -- (14.5,3.75)
	(15,3.25) -- (15,3.75);

\draw
	(10.2,3) node {{\color{red}$\bar Z_i$}};
\draw
	(14,4) node {{\color{blue}$\bar X_i$}};

\end{tikzpicture}
\hspace{.2cm}
\includegraphics[scale=.3]{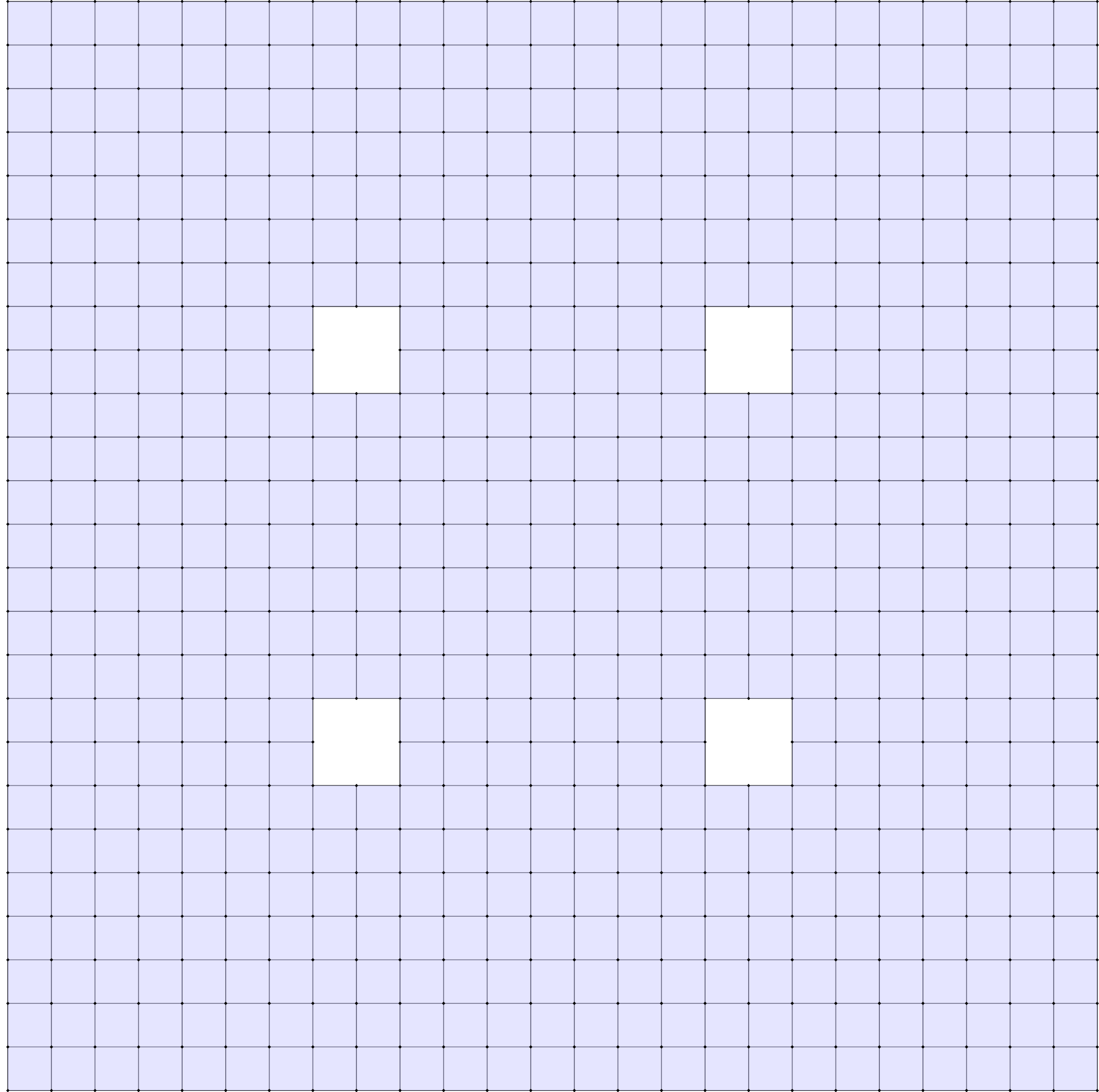}

\caption{Left: Structure of the square hole architecture with minimum distance $d = 4t$. 
Right: An example with 4 holes encoding k=4 qubits with minimum distance $d = 8$.}
\label{fig:square_holes}
\end{figure}

\medskip
Let us describe logical operations for these codes. 
Consider for instance a square hole surface code with $k$ holes.
As illustrated in Figure~\ref{fig:square_holes}, one can choose 
$\bar Z_i$ which is a $Z$-operator whose support is a cycle enclosing the $i$-th hole and 
$\bar X_i$ which is a path in the dual graph connecting this hole to the boundary
of the planar lattice.

\begin{example}
The square hole architecture $\Sq(h, t)$ yields a surface code whose
parameters satisfy
$$
n \sim 3kd^2
$$
when both $k$ and $d\rightarrow +\infty$.
\end{example}

The computation of these parameters relies on the following remark.
When puncturing, $2t^2 - 2t$ qubits are removed per hole.
There remains $n = 2 L^2 + 2 L - h^2 (2t^2-2t)$ physical qubits.
Indeed, a $(t \times t)$-square region of the square lattice contains $2t^2+2t$ edges,
its boundary contains $4t$ edges and its interior contains $2t^2-2t$ edges.
This proves that when $k, d \rightarrow +\infty$.
%\begin{align*}
%n 
%& = 2 L^2 + 2 L - h^2 (2t^2-2t)\\
%& = 2 ( h(5t-1)+(4t-1) )^2 + 2L - h^2 (2t^2-2t)\\
%& \sim 50 h^2 t^2 - 2h^2 t^2\\
%& \sim 48 k t^2\\
%& \sim 3 k d^2
%\end{align*}
\begin{align*}
n  = 2 L^2 + 2 L - h^2 (2t^2-2t) \sim 48 k t^2 \sim 3 k d^2 \cdot
\end{align*}
Therein, we used $k = h^2$ which is the number of holes and $t \sim d/4$ when $d$ diverges.

\subsection{Diamond hole architecture} \label{section:diamond_hole}

\begin{figure}[h]
\begin{tikzpicture}[scale=.35, every node/.style={scale=.4}]

%\tikzstyle{every node}=[circle, draw, fill=black!100, inner sep=1pt, minimum width=2pt]

%large rectangle
\draw
	(0,0) -- (15,15) -- (4,26) -- (-11,11) -- (0,0);
\fill[color=blue!20, opacity=.5]
	(0,0) -- (15,15) -- (4,26) -- (-11,11) -- (0,0);
	
%holes
\draw[xshift = 0cm, yshift = 6cm]
	(0,0) -- (1,1) -- (0,2) -- (-1,1) -- (0,0);
\fill[color=white, xshift = 0cm, yshift = 6cm]
	(0,0) -- (1,1) -- (0,2) -- (-1,1) -- (0,0);

\draw[xshift = 4cm, yshift = 10cm]
	(0,0) -- (1,1) -- (0,2) -- (-1,1) -- (0,0);
\fill[color=white, xshift = 4cm, yshift = 10cm]
	(0,0) -- (1,1) -- (0,2) -- (-1,1) -- (0,0);

\draw[xshift = 8cm, yshift = 14cm]
	(0,0) -- (1,1) -- (0,2) -- (-1,1) -- (0,0);
\fill[color=white, xshift = 8cm, yshift = 14cm]
	(0,0) -- (1,1) -- (0,2) -- (-1,1) -- (0,0);

\draw[xshift = -4cm, yshift = 10cm]
	(0,0) -- (1,1) -- (0,2) -- (-1,1) -- (0,0);
\fill[color=white, xshift = -4cm, yshift = 10cm]
	(0,0) -- (1,1) -- (0,2) -- (-1,1) -- (0,0);

\draw[xshift = 0cm, yshift = 14cm]
	(0,0) -- (1,1) -- (0,2) -- (-1,1) -- (0,0);
\fill[color=white, xshift = 0cm, yshift = 14cm]
	(0,0) -- (1,1) -- (0,2) -- (-1,1) -- (0,0);

\draw[xshift = 4cm, yshift = 18cm]
	(0,0) -- (1,1) -- (0,2) -- (-1,1) -- (0,0);
\fill[color=white, xshift = 4cm, yshift = 18cm]
	(0,0) -- (1,1) -- (0,2) -- (-1,1) -- (0,0);

%arrows
%inner arrow
\draw[<->, xshift = 4cm, yshift = 4cm]
	(0, 6.1) -- (0, 7.9) node [midway, right, above] {$2t-1$};
\draw[<->, xshift = 4cm, yshift = 4cm]
	(-.9, 7) -- (.9, 7);

%diagonal arrow north est
\draw[<->, xshift = 4cm, yshift=4cm]
	(-.5, 6.5) -- (-3.5, 3.5) node [midway, below right] {$d$};
\draw[<->, xshift = 8cm, yshift=8cm]
	(-.5, 6.5) -- (-3.5, 3.5) node [midway, below right] {$d$};

%diagonal arrow north west
\draw[<->]
	(3.5, 11.5) -- (0.5, 14.5) node [midway, above right] {$d$};
\draw[<->, xshift = 4cm, yshift=-4cm]
	(3.5, 11.5) -- (0.5, 14.5) node [midway, above right] {$d$};
	
%horizontal arrow
\draw[<->]
	(3, 11) -- (-3,11) node [midway, below right] {$d$};
\draw[<->, xshift = 8cm, yshift=0cm]
	(3, 11) -- (-3,11) node [midway, below right] {$d$};

%vertical arrow	
\draw[<->]
	(4, 10) -- (4,4) node [midway, right] {$d$};
\draw[<->, xshift = 0cm, yshift=8cm]
	(4, 10) -- (4,4) node [midway, right] {$d$};

%include examples
\node[inner sep=0pt] (ex2) at (18,21)
    {\includegraphics[width=.6\textwidth]{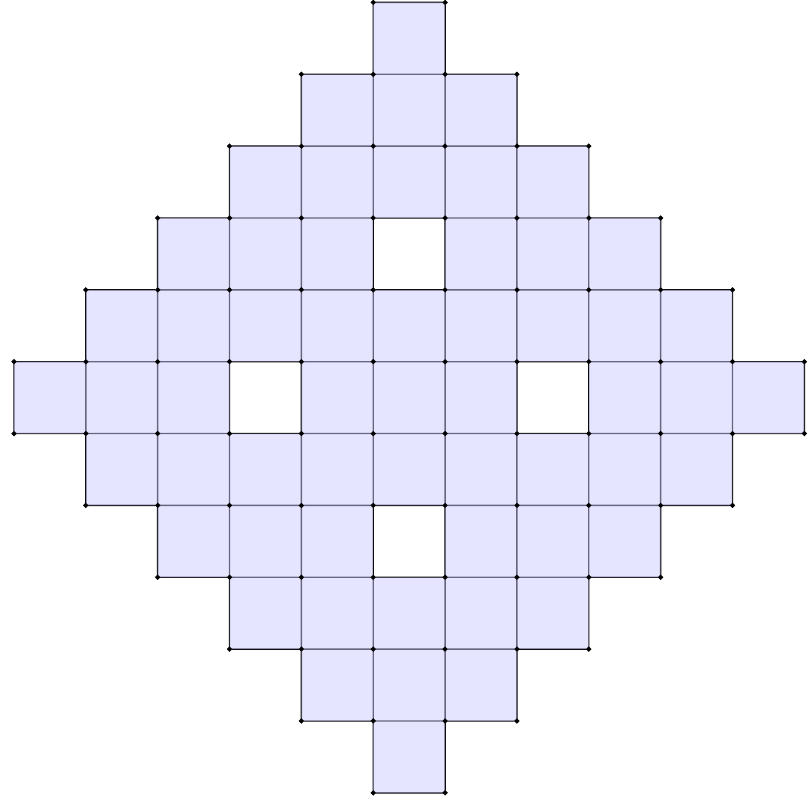}};
\node[inner sep=0pt] (ex1) at (16,6)
    {\includegraphics[width=.9\textwidth]{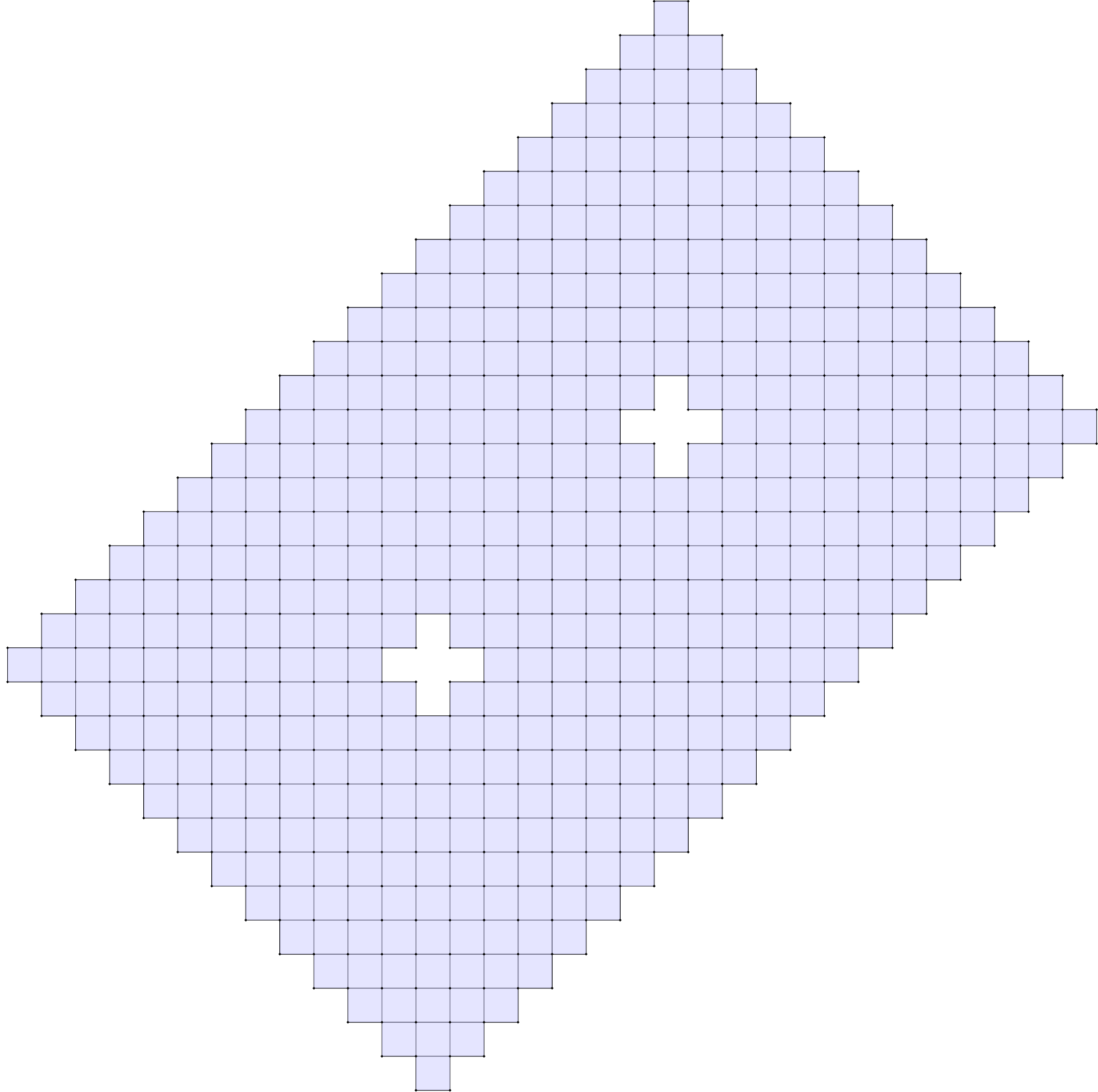}};
    
\end{tikzpicture}

\caption{Left: Representation of the diamond hole architecture 
with minimum distance $d = 4(2t-1)$ with the $X$-distance between holes and
the boundary.
Top right: A diamond-hole surface code with 4 holes encoding $k=4$ qubits with 
minimum distance $d = 4$. 
Bottom right: A diamond-hole surface code with 2 holes encoding $k=2$ qubits with minimum distance $d = 12$.}
\label{fig:diamond_holes}
\end{figure}
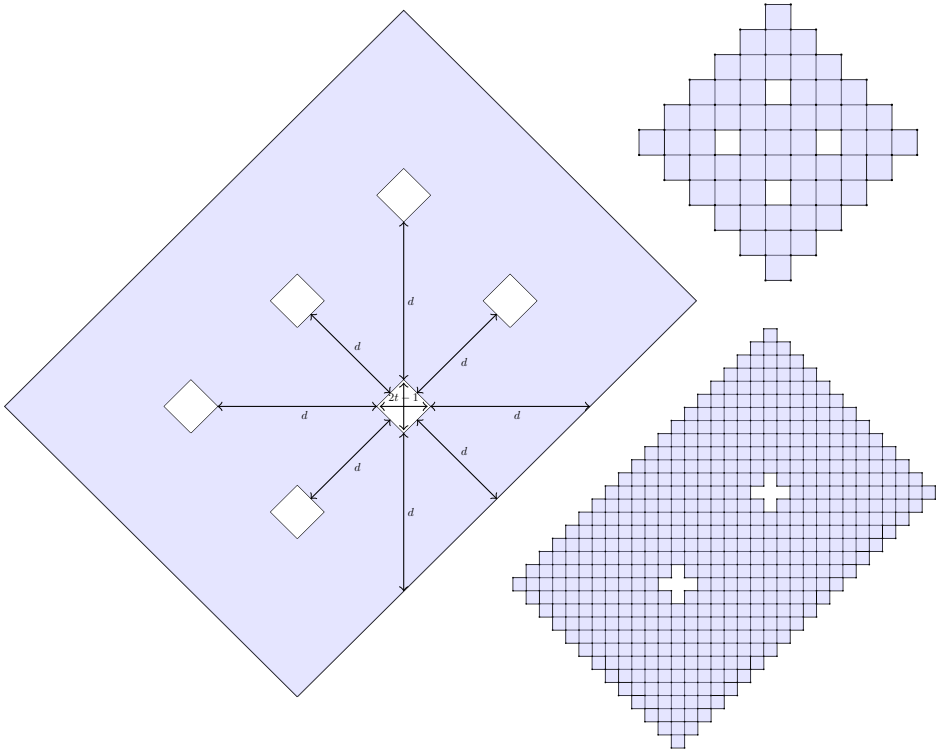

We can easily reduce the number of physical qubits required in the 
previous family of codes, for instance by simply truncating the corner. 
Indeed, the qubits at each of the 4 corners of the lattice are far away from
each hole. Cutting these corners will not reduce the minimum distance.
More generally, note that a hole in the square hole architecture is 
surrounded by only 4 closest holes (2 in the horizontal direction and 2 in the
vertical direction). We will construct a packing of these holes that requires 
much less resources. 
%This packing is similar to the rotated surface code 
%\cite{BM07:rotated_code}.

\medskip
In order to get some intuition on this packing problem, we introduce 
the {\em $X$-distance} between two holes which is defined as the minimum 
length of a path in the dual graph connecting these two holes. Such a path
is the support of a $X$-logical operator.
Given a punctured lattice, the $X$-distance between two distinct holes is 
at least $d_X$. Two holes separated by a $X$-distance exactly $d_X$ are 
called {\em tangent holes}. They cannot be closer.
For instance, in the square hole architecture each hole is tangent to up to 4 holes.
Based on the similarity with a sphere packing problem, we choose a punctured lattice 
such that each hole is tangent to a large number of holes.

\medskip
Let us now define the {\em diamond hole architecture} or {\em  diamond-hole surface codes}
in which each hole is tangent to up to 8 distinct holes.
It will be denoted $\D(h, h', t)$ or $\D(h, t)$ when $h = h'$,
where $h, h'$ and $t$ are non-negative integers.
This is also a surface codes based on a planar square lattice punctured by
holes with closed boundaries. It is represented in Figure~\ref{fig:diamond_holes}.
The parameter $t$ fixes the size of holes. For $t = 1$, each hole is
simply a single face. For $t = 2$ we extend this hole by 
adding the 4 faces incident to it. These 5 faces correspond to a vertex and 
its neighbours in the dual graph, that is a ball of radius 1 in the dual.
More generally, a hole is the dual of a ball of radius $t$ in the dual graph.
The minimum distance will be given by the perimeter of holes, that is 
$d = 4(2t-1)$.
A set of $h \times h'$ holes are punctured in a rotated square lattice as 
in Figure~\ref{fig:diamond_holes}.
To keep a minimum distance $d = 4(2t-1)$, holes are separated to one another and 
to the boundary of the lattice by a $X$-distance $d$.
It was already noticed that such a rotation increases the distance for  
toric codes (without punctures) \cite{BM06:optimal_toric_codes}.

Just as with square holes, logical operators are generated by the
$Z$-operators wrapping around holes and $X$-operators connecting holes
to the boundary in the dual graph.

\begin{example}
The diamond hole architecture $\D(h, t)$ yields a surface code whose
parameters satisfy
$$
n \sim1.5 k d^2
$$
when $k$ and $d\rightarrow \infty$.
\end{example}

Indeed, the total number of physical qubits in $D(h, t)$
is $n = 4(h(t+d/2) + d/2)^2 - 4 h^2 t^2$.
For this enumeration remark that, before puncturing, the rotated lattice 
contains $4 L_h^2$ edges. Each puncture removes $4t^2$ edges, leaving us with $4 L_h^2 - 4 h t^2$.
When both $k$ and $d$ diverge, we get
\begin{align*}
%n = 4(h(5d/8) + d/2)^2 - 4 h^2 (d/8)^2 \sim (100/64)^2 k d^2 - (4/64) k d^2 \sim (96/64)^2 k d^2
n = 4(h(5d/8) + d/2)^2 - 4 h^2 (d/8)^2 \sim (96/64)^2 k d^2
\end{align*}
using $t \sim d/8$ and $k = h^2$.
 
\medskip
One may wonder whether we can place these holes in such a way that a hole is tangent to
more than 8 holes.
This diamond hole lattice is a {\em perfect lattice} in the following sense.
For each qubit of the lattice, either there is a unique hole at $X$-distance less than $d/2$
or it is at distance less $d/2$ from the outer boundary of the lattice.
Stated differently, this proves that the regions $B(h, d/2)$ at $X$-distance $d/2$
from holes and from the boundary form a partition of the lattice.
Therein, the $X$-distance is extended to measure the distance between 2 edges as follows.
The $X$ distance between 2 edges $e$ and $f$ is the minimum length of a path of the dual graph
whose first edge is $e$ and last edge is $f$.
These regions are exactly the neighborhood of the holes considered
in the introduction of this section to prove the bound $k d^2 = O(n)$.
The term perfect lattice is chosen for the resemblance with perfect codes \cite{MS77}.
This argument does not exclude a better tradeoff for holes with a different shape.

\subsection{Mixed boundary diamond hole architecture}

In order to further improve the parameters of surface codes, we will
introduce open boundaries around every hole. Our basic
idea is to divide the boundary of a hole in an alternate sequence of open and closed
paths. This reduces the minimum distance. One can then shrink the lattice,
cutting the required number of physical qubits. Moreover, the number of
encoded qubits $k$ increases. Overall, we will prove that this transformation allows
us to achieve a better tradeoff. 

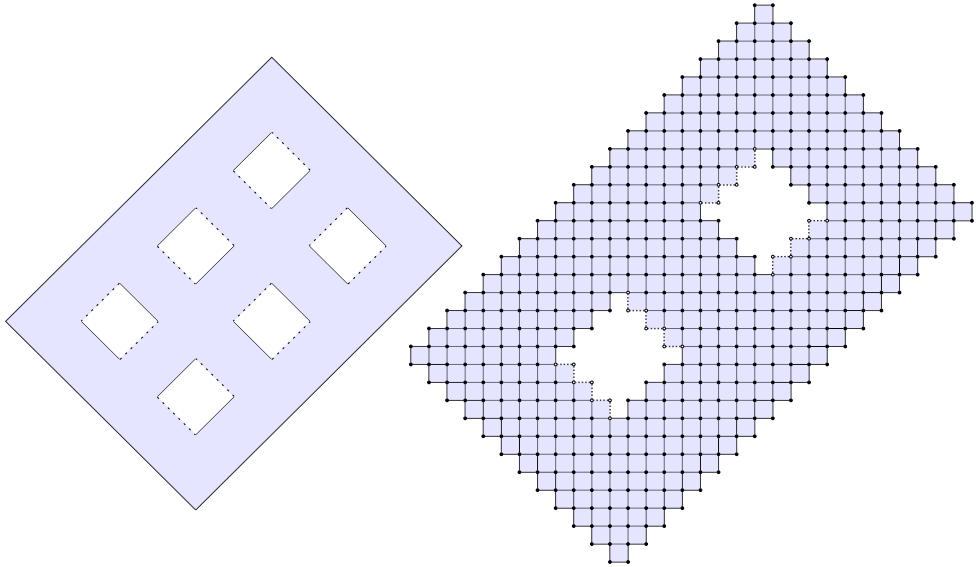
\begin{figure}[h]
\begin{center}
\begin{tikzpicture}[scale=.5, every node/.style={scale=.3}]

%\tikzstyle{every node}=[circle, draw, fill=black!100, inner sep=1pt, minimum width=2pt]

%large rectangle
\draw
	(0,0) -- (7,7) -- (2,12) -- (-5,5) -- (0,0);
\fill[color=blue!20, opacity=.5]
	(0,0) -- (7,7) -- (2,12) -- (-5,5) -- (0,0);
	
%holes
\draw[xshift = 0cm, yshift = 2cm]
	(0,0) -- (1,1)
	(0,2) -- (-1,1);
\draw[xshift = 0cm, yshift = 2cm, dotted, thick]
	(1,1) -- (0,2)
	(-1,1) -- (0,0);
\fill[color=white, xshift = 0cm, yshift = 2cm]
	(0,0) -- (1,1) -- (0,2) -- (-1,1) -- (0,0);

\draw[xshift = 2cm, yshift = 4cm, dotted, thick]
	(0,0) -- (1,1)
	(0,2) -- (-1,1);
\draw[xshift = 2cm, yshift = 4cm]
	(1,1) -- (0,2)
	(-1,1) -- (0,0);
\fill[color=white, xshift = 2cm, yshift = 4cm]
	(0,0) -- (1,1) -- (0,2) -- (-1,1) -- (0,0);

\draw[xshift = 4cm, yshift = 6cm, dotted, thick]
	(0,0) -- (1,1)
	(0,2) -- (-1,1);
\draw[xshift = 4cm, yshift = 6cm]
	(1,1) -- (0,2)
	(-1,1) -- (0,0);
\fill[color=white, xshift = 4cm, yshift = 6cm]
	(0,0) -- (1,1) -- (0,2) -- (-1,1) -- (0,0);

\draw[xshift = -2cm, yshift = 4cm, dotted, thick]
	(0,0) -- (1,1)
	(0,2) -- (-1,1);
\draw[xshift = -2cm, yshift = 4cm]
	(1,1) -- (0,2)
	(-1,1) -- (0,0);
\fill[color=white, xshift = -2cm, yshift = 4cm]
	(0,0) -- (1,1) -- (0,2) -- (-1,1) -- (0,0);

\draw[xshift = 0cm, yshift = 6cm]
	(0,0) -- (1,1)
	(0,2) -- (-1,1);
\draw[xshift = 0cm, yshift = 6cm, dotted, thick]
	(1,1) -- (0,2)
	(-1,1) -- (0,0);
\fill[color=white, xshift = 0cm, yshift = 6cm]
	(0,0) -- (1,1) -- (0,2) -- (-1,1) -- (0,0);

\draw[xshift = 2cm, yshift = 8cm]
	(0,0) -- (1,1)
	(0,2) -- (-1,1);
\draw[xshift = 2cm, yshift = 8cm, dotted, thick]
	(1,1) -- (0,2)
	(-1,1) -- (0,0);
\fill[color=white, xshift = 2cm, yshift = 8cm]
	(0,0) -- (1,1) -- (0,2) -- (-1,1) -- (0,0);

%include examples
\node[inner sep=0pt] (ex1) at (13,6)
    {\includegraphics[width=1.6\textwidth]{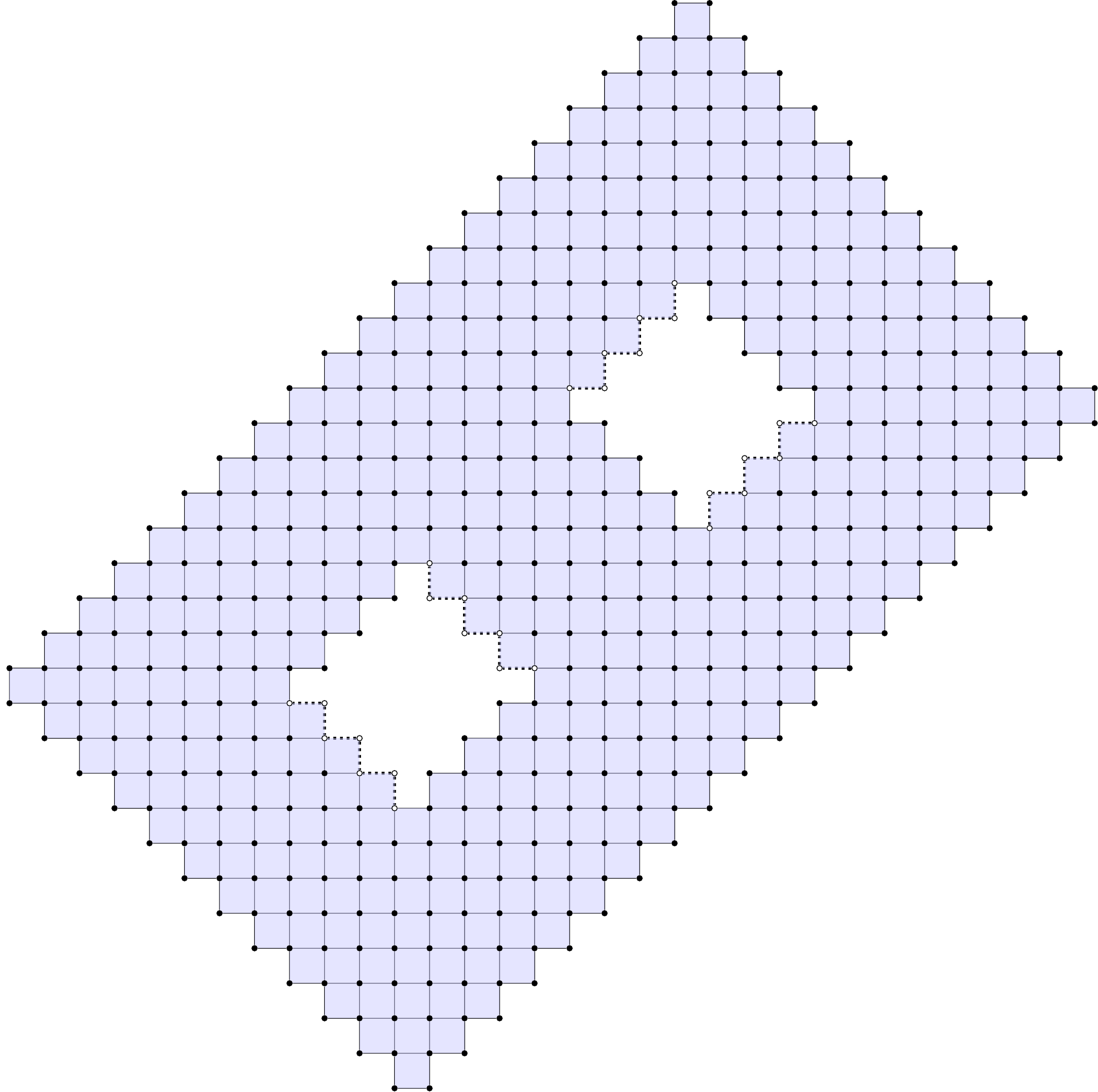}};
\end{tikzpicture}

\caption{Left: Structure of the mixed diamond hole architecture.
Dotted edges represent open boundaries. The distance between two holes
coincides with the size of holes.
Right: A mixed diamond-hole surface code with 2 holes encoding 
$k=5$ qubits with minimum distance $d = 8$.}
\label{fig:mixed_diamond_holes}
\end{center}
\end{figure}

We start with a lattice punctured with diamond holes like in 
Section~\ref{section:diamond_hole} and we will open some of the 
boundaries of the holes as shown in Figure~\ref{fig:mixed_diamond_holes}.

Let us define the {\em mixed diamond hole architecture} that we denote 
$\D_4(h, h', t)$ or $\D_4(h, t)$ when $h = h'$.
The underlying lattice is a rotated square lattice. It is punctured 
by diamond holes whose size depends on the parameter $t$ just 
as in the diamond hole architecture. Recall that a diamond hole is
the dual of a ball of radius $t$.
The edges on 2 opposite sites of these holes are then declared to be open.
The boundary of such a hole is an alternate sequence of 2 closed paths
and 2 open paths as we can see in Figure~\ref{fig:mixed_diamond_holes}.
More precisely, the perimeter of such a hole contains $4(2t-1)$ edges, 2 paths
of length $2t-2$ are declared to be open and 2 paths of length $2t$ are kept closed.
The minimum distance of a code based on such punctures is then at most 
$d = 2t$. We fix the $X$-distance between holes to be as small as possible and 
at least $2t$ in order to obtain a minimum distance $d = 2t$ for the code.
Open and closed sides of holes are chosen such that adjacent holes
face each other with different kinds of boundaries. This considerably reduces 
the number of low-weight errors.

\begin{example}
The mixed diamond hole architecture $\D_4(h, t)$ yields a surface code whose
parameters satisfy
$$
n \sim k d^2
$$
when $k$ and $d\rightarrow \infty$.
\end{example}

Adapting the arguments of the previous section, we see that the architecture 
$\D_4(h, t)$ requires a rotated square lattice of size $L_h \times L_h$
where $L_h = h(d/2+d/2)+d/2$ and that $h^2$ rotated square holes
of size $d/2 \times d/2$ are removed.
We obtain $n \sim 4 h^2 d^2 - 4 h^2 (d/2)^2 = 3 h^2 d^2 = kd^3$
since we have $k = 3h^2-1 \sim 3h^2$.

\section{Concluding remarks}

We have obtained a 3-fold reduction in the number of physical qubits compared to the 
most widely studied surface code architecture \cite{FMMC12}. Note that this improvement 
was achieved without affecting the minimum distance. In general, the performance 
of a code is dictated not only by the minimum distance but also by the number of 
distinct errors achieving this minimum distance, and more generally by the entire 
weight enumerator polynomial. This is why open and closed sides of holes were 
chosen such that neighbours holes face each other with different kinds of boundaries, 
considerably reducing the number of low-weight errors. The problem of optimizing 
a lattice geometry taking into account these combinatorial factors is in general very hard. 
In an upcoming paper, we will present a linear-time benchmarking algorithm which provides 
a quick way of numerically comparing different geometries.

\medskip
The naive bound explained in the introduction of Section~\ref{section:planar_codes} proves that 
$kd^2 \leq n$ for planar square lattice architecture based on closed holes.
We believe that the diamond-hole architecture provides the best tradeoff
for planar Euclidean surface codes with closed boundaries, so the packing argument 
presented in the introduction of this section can probably be refined to prove that 
$kd^2 \leq 1.5n$ for such codes. We have shown that surface codes with mixed boundaries 
can violate this bound. To encode a single logical qubit, the rotated surface
code \cite{BM06:optimal_toric_codes, HFDV12:surface_code_surgery, TS14:small_surface_codes,
GCS15:quantum_architecture} remains the best alternative.

\medskip
We exploited our formalism to construct surface codes with better parameters. Optimizing the
parameters of surface codes and in particular the minimum distance naturally leads to a better
performance for a depolarizing noise. However, a different type noise may require a different strategy.
For instance, to fight the effect of an asymmetric Pauli noise with a high $Z$-error probability 
and a low $X$-error probability, one should consider asymmetric surface codes with a much stronger
error-correction capability against $Z$-errors than $X$-errors \cite{DT14:correlations}.
More generally, it is crucial to understand the behaviour of generalized surface codes under more
general non-Pauli noise.

%\medskip
%While we focused here on surface codes for qubits, a.k.a., $\mathbb Z_2$ quantum doubles, 
%our construction can be extended to $\mathbb Z_d$ quantum doubles \cite{Ki03} in the standard way.

\medskip
{\bf Acknowledgement:} 
The authors would like to thank Aleksander Kubica and Fernando Pastawski for their
comments on a preliminary version of this article.

ND was supported by the U.S. Army Research Office under Grant 
No. W911NF-14-1-0272 and by the NSF under Grant No. PHY-1416578,
ND acknowledges funding provided by the Institute for Quantum Information and Matter, 
an NSF Physics Frontiers Center (NSF Grant PHY-1125565) with support of 
the Gordon and Betty Moore Foundation (GBMF-2644).
PI and DP were supported by Canada's NSERC and by the Canadian Institute 
for Advanced Research

%\bibliographystyle{unsrt}
%\bibliography{biblio.bib}

\begin{thebibliography}{10}

\bibitem{Ki03}
A.~Y. Kitaev.
\newblock Fault-tolerant quantum computation by anyons.
\newblock {\em Annals of Physics}, 303(1):27, 2003.

\bibitem{FM01}
M.~H. Freedman and D.~A. Meyer.
\newblock Projective plane and planar quantum codes.
\newblock {\em Foundations of Computational Mathematics}, 1(3):325--332, 2001.

\bibitem{BK98}
S.~B. Bravyi and A.~Y. Kitaev.
\newblock Quantum codes on a lattice with boundary.
\newblock arXiv:9811052, 1998.

\bibitem{DKLP02}
E.~Dennis, A.~Kitaev, A.~Landahl, and J.~Preskill.
\newblock Topological quantum memory.
\newblock {\em Journal of Mathematical Physics}, 43:4452, 2002.

\bibitem{RH07}
R.~Raussendorf and J.~Harrington.
\newblock Fault-tolerant quantum computation with high threshold in two
  dimensions.
\newblock {\em Physical Review Letters}, 98(19):190504, 2007.

\bibitem{RHG07}
R.~Raussendorf, J.~Harrington, and K.~Goyal.
\newblock Topological fault-tolerance in cluster state quantum computation.
\newblock {\em New Journal of Physics}, 9:199, 2007.

\bibitem{RHG06}
R.~Raussendorf, J.~Harrington, and K.~Goyal.
\newblock A fault-tolerant one-way quantum computer.
\newblock {\em Annals of Physics}, 321(9):2242 -- 2270, 2006.

\bibitem{WFHL11}
D.~S. Wang, A.~G. Fowler, and L.~C.~L. Hollenberg.
\newblock Surface code quantum computing with error rates over 1%.
\newblock {\em Phys. Rev. A}, 83:020302, Feb 2011.

\bibitem{FMMC12}
A.~G. Fowler, M.~Mariantoni, J.~M. Martinis, and A.~N. Cleland.
\newblock Surface codes: Towards practical large-scale quantum computation.
\newblock {\em Physical Review A}, 86(3):032324, 2012.

\bibitem{HFDV12:surface_code_surgery}
C.~Horsman, A.~G. Fowler, S.~Devitt, and R.~Van~Meter.
\newblock Surface code quantum computing by lattice surgery.
\newblock {\em New Journal of Physics}, 14(12):123011, 2012.

\bibitem{BT15:hyperbolic}
N.~P. Breuckmann and B.~M. Terhal.
\newblock Constructions and noise threshold of hyperbolic surface codes.
\newblock {\em arXiv preprint arXiv:1506.04029}, 2015.

\bibitem{Go97}
D.~Gottesman.
\newblock {\em Stabilizer Codes and Quantum Error Correction}.
\newblock PhD thesis, California Institute of Technology, 1997.

\bibitem{CS96}
A.~R. Calderbank and P.~W. Shor.
\newblock Good quantum error-correcting codes exist.
\newblock {\em Physical Review A}, 54(2):1098, 1996.

\bibitem{St96}
A.~Steane.
\newblock Multiple-particle interference and quantum error correction.
\newblock {\em Proc. of the Royal Society of London. Series A: Mathematical,
  Physical and Engineering Sciences}, 452(1954):2551--2577, 1996.

\bibitem{KYP15}
A.~Kubica, B.~Yoshida, and F.~Pastawski.
\newblock Unfolding the colour code.
\newblock {\em New Journal of Physics}, 17(8):083026, 2015.

\bibitem{Ha02}
A.~Hatcher.
\newblock {\em Algebraic topology}.
\newblock Cambridge University Press, 2002.

\bibitem{Gi10}
P.~Giblin.
\newblock {\em Graphs, surfaces and homology}.
\newblock Cambridge University Press, 2010.

\bibitem{Be73}
C.~Berge.
\newblock {\em Graphs and Hypergraphs}.
\newblock Elsevier, 1973.

\bibitem{Bo12}
B.~Bollobas.
\newblock {\em Graph theory: an introductory course}, volume~63.
\newblock Springer Science \& Business Media, 2012.


\bibitem{BPT10a}
S. Bravyi, D. Poulin, and B.M. Terhal.
\newblock Tradeoffs for reliable quantum information storage in {2D} systems.
\newblock {\em Physical Review Letters} 104:050503, 2010.


\bibitem{FML02}
M.~H. Freedman, D.~A. Meyer, and F.~Luo.
\newblock Z2-systolic freedom and quantum codes.
\newblock {\em Mathematics of Quantum Computation, Chapman \& Hall/CRC}, pages
  287--320, 2002.

\bibitem{Ze09}
G.~Z{\'e}mor.
\newblock On {C}ayley graphs, surface codes, and the limits of homological
  coding for quantum error correction.
\newblock In {\em Proc. of the 2nd International Workshop on Coding and
  Cryptology, IWCC 2009}, pages 259--273. Springer-Verlag, 2009.

\bibitem{DZ10}
N.~Delfosse and G.~Z\'emor.
\newblock Quantum erasure-correcting codes and percolation on regular tilings
  of the hyperbolic plane.
\newblock In {\em Proc. of IEEE Information Theory Workshop, ITW 2010}, pages
  1--5, 2010.

\bibitem{De13:tradeoffs}
N.~Delfosse.
\newblock  A decoding algorithm for CSS codes using the X/Z correlations.
\newblock In {\em Proc. of IEEE International Symposium on Information Theory,
ISIT 2014}, pages 1071 - 1075, 2014.

\bibitem{DT14:correlations}
N.~Delfosse and J.-P.~Tillich
\newblock 
\newblock In {\em Proc. of IEEE International Symposium on Information Theory,
  ISIT 2013}, pages 917--921, 2013.

\bibitem{BM06:optimal_toric_codes}
H.~Bombin and M.~A. Martin-Delgado.
\newblock Topological quantum error correction with optimal encoding rate.
\newblock {\em Physical Review A}, 73(6):062303, 2006.

\bibitem{MS77}
F.~J. MacWilliams and N.~J.~A. Sloane.
\newblock {\em The theory of error correcting codes}.
\newblock North-Holland mathematical library. North-Holland Pub. Co. New York,
  Amsterdam, New York, 1977.
\newblock Includes index.

\bibitem{TS14:small_surface_codes}
Y.~Tomita and K.~M. Svore.
\newblock Low-distance surface codes under realistic quantum noise.
\newblock {\em Phys. Rev. A}, 90:062320, Dec 2014.

\bibitem{GCS15:quantum_architecture}
J.~M. Gambetta, J.~M. Chow, and M.~Steffen.
\newblock Building logical qubits in a superconducting quantum computing
  system.
\newblock {\em arXiv preprint arXiv:1510.04375}, 2015.

\end{thebibliography}

\newcommand{\SortNoop}[1]{}

\end{document}